\documentclass{article}%
\usepackage{amsmath}
\usepackage{amsfonts}
\usepackage{amssymb}
\usepackage{graphicx}
\usepackage[colorlinks, linkcolor={blue}, urlcolor={blue}, bookmarks={false},
pdfpagelayout={SinglePage}]{hyperref}
\usepackage[UKenglish]{isodate}
\usepackage{algorithm}
\usepackage{algpseudocode}
\usepackage{lscape}%
\newtheorem{theorem}{Theorem}

\newtheorem{lemma}[theorem]{Lemma}
\newtheorem{notation}[theorem]{Notation}

\newtheorem{remark}[theorem]{Remark}

\newenvironment{proof}[1][Proof]{\noindent\textbf{#1.} }{\rule{0.5em}{0.5em}}
\begin{document}

\title{Accelerated Quantum Amplitude Estimation without QFT}
\author{Alet Roux\thanks{\texttt{alet.roux@york.ac.uk}; Department of Mathematics,
University of York, Heslington, York YO10~5DD, United Kingdom.} and Tomasz
Zastawniak\thanks{\texttt{tomasz.zastawniak@york.ac.uk}; Department of
Mathematics, University of York, Heslington, York YO10~5DD, United Kingdom.}}
\date{}
\maketitle

\begin{abstract}
We put forward a Quantum Amplitude Estimation algorithm delivering superior
performance (lower quantum computational complexity and faster classical
computation parts) compared to the approaches available to-date. The algorithm
does not relay on the Quantum Fourier Transform and its quantum computational
complexity is of order $O(\frac{1}{\varepsilon})$ in terms of the target
accuracy $\varepsilon>0$. The $O(\frac{1}{\varepsilon})$ bound on quantum
computational complexity is also superior compared to those in the earlier
approaches due to smaller constants. Moreover, a much tighter bound is
obtained by means of computer-assisted estimates for the expected
value of quantum computational complexity. The correctness of the algorithm
and the $O(\frac{1}{\varepsilon})$ bound on quantum computational complexity
are supported by precise proofs.

\end{abstract}

\section{Introduction}

Let $\mathbf{A}$ be a unitary operator representing a quantum circuit such
that%
\begin{equation}
\mathbf{A}\left\vert 0\right\rangle _{w+1}=\sqrt{1-a}\left\vert \Psi
_{0}\right\rangle _{w}\!\left\vert 0\right\rangle +\sqrt{a}\left\vert \Psi
_{1}\right\rangle _{w}\!\left\vert 1\right\rangle \label{Eq:a43ff1}%
\end{equation}
for some unknown $a\in\lbrack0,1]$, where $\left\vert \Psi_{0}\right\rangle
_{w}$ and $\left\vert \Psi_{1}\right\rangle _{w}$ are normalised states of
width~$w$ (that is, $w$-qubit states). The goal of Quantum Amplitude
Estimation (QAE) is to compute an estimate~$\hat{a}$ of~$a$ to within given
accuracy $\varepsilon>0$ at a prescribed confidence level $1-\alpha\in\left(
0,1\right)  $, so that%
\begin{equation}
\mathbb{P}\{|\hat{a}-a|\leq\varepsilon\}\geq1-\alpha. \label{Eq:nf8as6g}%
\end{equation}

QAE is one of the fundamental procedures in quantum computing, and a building
block for quantum algorithms in diverse areas such as machine learning,
chemistry, or finance. In particular, QAE gives rise to quadratic quantum
speedup in Monte Carlo estimation.

Brassard \emph{et al.}~\cite{Bra2002} were the first to establish a QAE
procedure by combining ideas from the Grover and Shor quantum algorithms. The
approach of~\cite{Bra2002}, when adapted to an operator $\mathbf{A}$ of the
form~(\ref{Eq:a43ff1}), was to consider the operator%
\[
\mathbf{Q}=\mathbf{AS}_{0}\mathbf{A}^{\dagger}\mathbf{S}_{\Psi_{0}},
\]
where $\mathbf{S}_{0}=\mathbf{I}_{w+1}-2\left\vert 0\right\rangle
_{w+1}\!\left\langle 0\right\vert _{w+1}$ and $\mathbf{S}_{\Psi_{0}%
}=\mathbf{I}_{w}\otimes\left(  \mathbf{I}-2\left\vert 0\right\rangle
\!\left\langle 0\right\vert \right)  $, and apply to the state $\mathbf{A}%
\left\vert 0\right\rangle _{w+1}$ the operators $\mathbf{Q}^{2^{0}}%
,\mathbf{Q}^{2^{1}},\ldots,\mathbf{Q}^{2^{j-1}}$ controlled by~$j$ ancillary
qubits, sandwiched between the Quantum Fourier Transform (QFT) and the inverse
QFT acting on the ancillary qubits. The quantum computational complexity of
the QAE algorithm in~\cite{Bra2002}, understood as the number of applications
of~$\mathbf{Q}$, is of order~$O(\frac{1}{\varepsilon})$. However, the cost in
terms of quantum computing resources is considerable due to the use of QFT and
controlled gates $\mathbf{Q}^{2^{0}},\mathbf{Q}^{2^{1}},\ldots,\mathbf{Q}%
^{2^{j-1}}$. It is important to seek more efficient and less costly QAE
procedures that lend themselves to implementation on near-term quantum
computers, while matching the computational complexity of order~$O(\frac
{1}{\varepsilon})$.

Suzuki \emph{et al.}\textbf{~}\cite{Suz2020} put forward a version of QAE
which does not rely on QFT or controlled gates $\mathbf{Q}^{2^{0}}%
,\mathbf{Q}^{2^{1}},\ldots,\mathbf{Q}^{2^{j-1}}$, reducing the circuit width
and depth as compared to~\cite{Bra2002}. After measuring the rightmost qubit
in the states $\mathbf{Q}^{2^{0}}\mathbf{A}\left\vert 0\right\rangle
_{w+1},\mathbf{Q}^{2^{1}}\mathbf{A}\left\vert 0\right\rangle _{w+1}%
,\ldots,\mathbf{Q}^{2^{j-1}}\mathbf{A}\left\vert 0\right\rangle _{w+1}$, the
authors of~\cite{Suz2020} applied maximum likelihood estimation to obtain an
approximation of~$a$. They justified this procedure by heuristic
considerations and demonstrated empirically that the number of applications
of~$\mathbf{Q}$ needed to achieve accuracy~$\varepsilon$ appears to scale
roughly as $O\left(  \frac{1}{\varepsilon}\right)  $, but provided no rigorous
proof to support this conjecture.

Aaranson and Rall~\cite{Aar2020} proposed a QAE algorithm without QFT or
controlled gates~$\mathbf{Q}^{k}$, proven to achieve the same quantum
computational complexity of order $O(\frac{1}{\varepsilon})$ as
in~\cite{Bra2002}, but with large constants making it ill-suited for
implementation on near-term quantum computers. Near-term efficiency was
considerably improved by Grinko \emph{et al.}~\cite{Gri2021} at the expense of
asymptotic computational complexity of the QAE algorithm, for which the
authors obtained a bound of order $O(\frac{1}{\varepsilon}\log(\log\frac
{1}{\varepsilon}))$. A different algorithm of this kind, also with query
complexity of order $O(\frac{1}{\varepsilon}\log(\log\frac{1}{\varepsilon}))$,
belongs to Nakaji~\cite{Nak2020}. Fukuzawa \emph{et al.}~\cite{Fuk2023}
proposed a modification of the algorithm in~\cite{Gri2021} with computational
complexity of order~$O(\frac{1}{\varepsilon})$ and constants competitive with
those in~\cite{Gri2021} and \cite{Nak2020}.

Here we present another version of QAE, also related to that in~\cite{Gri2021}%
, but with computational complexity of order~$O(\frac{1}{\varepsilon})$ and
even better constants in the asymptotic bound, and better performance than all
the aforesaid approaches. In
Algorithms~\ref{algQAE} and~\ref{algQAEacc} we present variants of QAE which
do not utilize~QFT or controlled gates~$\mathbf{Q}^{k}$ either and accomplish
these aims. Theorems~\ref{Thm:2and6at3l} and~\ref{Thm:p9a3n7gd} provide
precise proofs of the correctness of the algorithms, that is, of
achieving~(\ref{Eq:nf8as6g}), with upper bounds on the number of applications
of~$\mathbf{Q}$ of order $O(\frac{1}{\varepsilon})$ and constants smaller than
in the papers listed above. Indeed, our algorithms, and particularly
Algorithm~\ref{algQAEacc}, reduce the number of applications of~$\mathbf{Q}$
as compared to all previous work, including~\cite{Gri2021} and~\cite{Fuk2023}. This being so, our
algorithms are\ well suited for near-term quantum computers. The classical
computation parts of the algorithms are also more efficient
compared, for example, to~\cite{IQAE} and~\cite{MIQAE}.

We rely on the following well-known property of~$\mathbf{Q}$ (see
\cite{Bra2002}):%
\begin{equation}
\mathbf{Q}^{k}\mathbf{A}\left\vert 0\right\rangle _{w+1}=\cos\left(  \left(
2k+1\right)  \theta\right)  \left\vert \Psi_{0}\right\rangle _{w}\!\left\vert
0\right\rangle +\sin\left(  \left(  2k+1\right)  \theta\right)  \left\vert
\Psi_{1}\right\rangle _{w}\!\left\vert 1\right\rangle \label{Eq:nnfa53r}%
\end{equation}
for each $k=0,1,2,\ldots$\thinspace, where%
\[
\sin^{2}(\theta)=a\quad\text{with}\quad\theta\in\textstyle\left[  0,\frac{\pi
}{2}\right]  .
\]
It means that a measurement of the rightmost qubit in~(\ref{Eq:nnfa53r}) will
produce~$\left\vert 1\right\rangle $ with probability $\sin^{2}\left(  \left(
2k+1\right)  \theta\right)  $. We write%
\[
n=%
%TCIMACRO{\TeXButton{QuantCirc}{\texttt{\textit{QuantCirc}}}}%
%BeginExpansion
\texttt{\textit{QuantCirc}}%
%EndExpansion
(k,N)
\]
for the number of times when outcome~$\left\vert 1\right\rangle $ is produced
in~$N$ runs of the circuit $\mathbf{Q}^{k}\mathbf{A}\left\vert 0\right\rangle
_{w+1}$ followed by a measurement of the rightmost
qubit. Hence,~$n$ can be regarded as the number of successful outcomes in~$N$
i.i.d.\ Bernoulli trials with probability of success $\sin^{2}\left(  \left(
2k+1\right)  \theta\right)  $, and $\frac{n}{N}$ can be taken as an estimate
for the probability of a successful outcome, that is, outcome~$\left\vert
1\right\rangle $.

Let%
\begin{equation}
\sin^{2}((2k+1)\hat{\theta})=\frac{n}{N} \label{Eq:bftarm9amn}%
\end{equation}
for some $\hat{\theta}\in\left[  0,\frac{\pi}{2}\right]  $, and let $E$ and
$F$ be given by~(\ref{Eq:nnc74tsfa}) and~(\ref{Eq:bc645qfsaa}). These
constants are such that if%
\[
|\sin^{2}((2k+1)\hat{\theta})-\sin^{2}((2k+1)\theta)|\leq 2E,
\]
(which will be so with high probability if~$N$ is large enough), then%
\[
|(2k+1)\hat{\theta}-(2k+1)\theta|\leq 2F,
\]
as long as there is an integer~$m$ such that both $(2k+1)\hat{\theta}$ and
$(2k+1)\theta$ belong to the same interval\ $\left[  m\frac{\pi}{2}%
,(m+1)\frac{\pi}{2}\right]  $. In that case we get $|\hat{\theta}-\theta
|\leq\frac{2F}{2k+1}$, and taking $\hat{a}=\sin^{2}(\hat{\theta})$ as an
estimate for $a=\sin^{2}(\theta)$, we obtain
\[
|\hat{a}-a|=|\sin^{2}(\hat{\theta})-\sin^{2}(\theta)|\leq|\hat{\theta}%
-\theta|\leq\frac{2F}{2k+1}.
\]
The accuracy $\frac{2F}{2k+1}$ of this estimate can be made as small as desired
by taking a sufficiently large~$k$. However, for this to work we need to know
the value of~$m$ such that $(2k+1)\theta$ belongs to $\left[  m\frac{\pi}%
{2},(m+1)\frac{\pi}{2}\right]  $, so that among all the solutions~$\hat
{\theta}$ of equation~(\ref{Eq:bftarm9amn}) we can select the one that belongs
to the interval $\left[  m\frac{\pi}{2},(m+1)\frac{\pi}{2}\right]  $. But we
do not know~$\theta$ (because we do not know~$a$), and will need to rely on
Lemma~\ref{Lem:bddta7d7gs} to keep track of~$m$. We do this in a different and
perhaps simpler manner than Grinko \emph{et al.}~\cite{Gri2021} or Fukuzawa
\emph{et al.}~\cite{Fuk2023}.

Where our approach differs significantly from~\cite{Gri2021} is in how we
achieve computational complexity $O(\frac{1}{\varepsilon})$ rather than
$O(\frac{1}{\varepsilon}\log(\log(\frac{1}{\varepsilon})))$. Instead of
assigning the same fraction~$\frac{\alpha}{T}$ of~$\alpha$ to each round~$i$
of the main loop in the algorithm in~\cite{Gri2021} (where $1-\alpha$ is the
prescribed confidence level and~$T$ is an upper bound on the number of
rounds), we take~$\alpha_{i}$ increasing with~$i$ as defined
by~(\ref{Eq:nc746ardf5}) or~(\ref{Eq:57ahd8anvx}) in such a manner that the
sum of the$~\alpha_{i}$'s over all the rounds~$i$ does not exceed~$\alpha$.
This is similar to, but different from \cite{Fuk2023}, leading to better
constants in the asymptotics for computational complexity and better
performance for relevant input data ranges. For low values of~$i$, that is,
when the powers of the $\mathbf{Q}$ operator in the quantum circuit are also
low, the lower values of~$\alpha_{i}$ force more runs of the circuit. But
fewer runs are needed as~$i$ and $\alpha_{i}$ increase and the powers
of\textbf{~}$\mathbf{Q}$ become higher. This is what makes it possible to
reduce the upper bound for the overall number of applications of~$\mathbf{Q}$
from $O(\frac{1}{\varepsilon}\log(\log(\frac{1}{\varepsilon})))$ to
$O(\frac{1}{\varepsilon})$ with tight constants, as demonstrated in
Theorems~\ref{Thm:2and6at3l} and~\ref{Thm:p9a3n7gd}. Moreover, in
Remarks~\ref{Rem:bbdta54v} and~\ref{Rem:awd6a54a} we show that our choice of
the\ $\alpha_{i}$ is optimal in a certain sense.

The above upper bounds on quantum computational complexity, which is a random
variable, are worst-case estimates holding with probability greater than or
equal to the prescribed confidence level $1-\alpha$. A more representative
quantity is the expected value of quantum computational complexity, that is,
the expected number of applications of $\mathbf{Q}$, for which we obtain a
much tighter bound by means of computer-assisted estimates.

\begin{notation}
\upshape Throughout this paper we denote by $\left[  x,y\right]  $ the closed
interval with end-points $x,y\in\mathbb{R}$, irrespective of whether $x\leq y$
or $y\leq x$.
\end{notation}

\section{Accelerated QAE algorithm}

We begin with formulating a simple version of our algorithm; see
Algorithm~\ref{algQAE}. We prove the correctness of this algorithm and obtain
an upper bound for quantum computational complexity in
Theorem~\ref{Thm:2and6at3l}.%

\begin{algorithm}
\caption{}\label{algQAE}
%\begin{algorithmic}[0] %without line numbers
\begin{algorithmic}[1] %with line numbers
\Require $\varepsilon>0$ and $\alpha\in(0,1)$
\State $E:=\frac{1}{2}\left(\sin^2\left(\frac{3}{7}\frac{\pi}{2}\right)
    -\sin^2\left(\frac{1}{3}\frac{\pi}{2}\right)\right),\quad F:=\frac{1}{2}\arcsin\sqrt{2E}$
\State $C:=\frac{4}{6F+\pi}$
\State $i:=-1,\quad K_0:=1,\quad m_0:=0$
\Repeat
    \State $i:=i+1$
    \State $\alpha_i:=C\alpha\varepsilon K_i,\quad
        N_i:=\left\lceil\frac{1}{2E^2}\ln\frac{2}{\alpha_i}\right\rceil$
    \State $n_i:=\texttt{\textit{QuantCirc}}(\frac{K_i-1}{2},N_i),\quad \hat{A}_i:=\frac{n_i}{N_i}$
    \State compute $\hat{\theta}^{\flat}_i$ and $\hat{\theta}^{\sharp}_i$ such that
        \Statex \quad\quad\quad
        $K_i\hat{\theta}^{\flat}_i,K_i\hat{\theta}^{\sharp}_i\in
        \big[m_i\frac{\pi}{2},(m_i+1)\frac{\pi}{2}\big]$
        \Statex \quad\quad\quad and\quad
        $\sin^2(K_i\hat{\theta}^{\flat}_i)=\max\{\hat{A}_i-E,0\}$
        \Statex \quad\quad\quad and\quad
        $\sin^2(K_i\hat{\theta}^{\sharp}_i)=\min\{\hat{A}_i+E,1\}$
    \State find $L_i\in\{3,5,7\}$ and $m_{i+1}\in L_im_i+\{0,\ldots,L_i-1\}$ such that
        \Statex \quad\quad\quad
        $[L_iK_i\hat{\theta}^{\flat}_i,L_iK_i\hat{\theta}^{\sharp}_i]\subset
        \big[m_{i+1}\frac{\pi}{2},(m_{i+1}+1)\frac{\pi}{2}\big]$
    \State $K_{i+1}:=L_iK_i$
\Until{$|\hat{\theta}^{\sharp}_i-\hat{\theta}^{\flat}_i|\le2\varepsilon$}
\State $I:=i$
\State \textbf{return} $\hat{a}:=\sin^2(\frac{1}2{}(\hat{\theta}^{\flat}_I+\hat{\theta}^{\sharp}_I))$
\end{algorithmic}
\end{algorithm}%

\begin{theorem}
\label{Thm:2and6at3l}\upshape Given an accuracy $\varepsilon>0$ and confidence
level $1-\alpha\in(0,1)$, Algorithm~\ref{algQAE} returns an estimate~$\hat{a}$
of~$a$ such that%
\begin{equation}
\mathbb{P}\{\left\vert \hat{a}-a\right\vert \leq\varepsilon\}\geq1-\alpha.
\label{Eq:ahs74ta5}%
\end{equation}
The algorithm involves at most%
\begin{align}
M  &  <\frac{1}{\varepsilon}\left(  \frac{1}{4E^{2}}\left(  \frac{3}{2}%
F\ln\left(  \frac{2}{C\alpha F}\right)  +\frac{3}{4}F\ln\left(  3\right)
+\frac{\pi}{4}\ln\left(  \frac{8}{C\alpha\pi}\right)  \right)  +\frac{3}%
{4}F+\frac{\pi}{8}\right) \nonumber\\
&  \approx\frac{1}{\varepsilon}\left(  85.637-55.674\ln\left(  \alpha\right)
\right)  \label{Eq:bbc65a43}%
\end{align}
applications of~$\mathbf{Q}$, where the constants $E,F,C$ are given by
(\ref{Eq:nnc74tsfa}), (\ref{Eq:bc645qfsaa}), (\ref{Eq:7ag45ar}). That is, the
quantum computational complexity~$M$ of the algorithm scales as~$O(\frac
{1}{\varepsilon})$.
\end{theorem}

\begin{proof}
Take an $a\in\left[  0,1\right]  $ and the angle $\theta\in\left[  0,\frac
{\pi}{2}\right]  $ such that%
\[
a=\sin^{2}(\theta),
\]
and for each $i=0,1,\ldots$ put%
\[
A_{i}=\sin^{2}(K_{i}\theta).
\]

For the sake of this argument, suppose that the \textbf{repeat} loop in
Algorithm~\ref{algQAE} does not break when the condition $|\hat{\theta}%
_{i}^{\sharp}-\hat{\theta}_{i}^{\flat}|\leq2\varepsilon$ is satisfied, but
keeps running indefinitely. This defines $\alpha_{i},N_{i},n_{i},\hat{A}%
_{i},K_{i},\hat{\theta}_{i}^{\sharp},\hat{\theta}_{i}^{\sharp},L_{i},m_{i}$
for each $i=0,1,\ldots$~. Then, $I$ can be taken to be the
smallest~$i=0,1,\ldots$ satisfying the condition $|\hat{\theta}_{i}^{\sharp
}-\hat{\theta}_{i}^{\flat}|\leq2\varepsilon$.

We have%
\begin{equation}
\hat{A}_{i}=\frac{n_{i}}{N_{i}}, \label{Eq:xa10an01}%
\end{equation}
where%
\begin{equation}
n_{i}=\textstyle%
%TCIMACRO{\TeXButton{QuantCirc}{\texttt{\textit{QuantCirc}}}}%
%BeginExpansion
\texttt{\textit{QuantCirc}}%
%EndExpansion
(\frac{K_{i}-1}{2},N_{i}) \label{Eq:yg5av16a}%
\end{equation}
can be regarded as the number of successful outcomes in~$N_{i}$
i.i.d.\ Bernoulli trials with probability of success~$A_{i}$. Let%
\[
B=\textstyle\bigcap_{i=0}^{I}B_{i},
\]
where, for each $i=0,1,\ldots\,$,
\[
B_{i}=\{|\hat{A}_{i}-A_{i}|\leq E\}
\]
with%
\begin{equation}
N_{i}=\left\lceil \frac{1}{2E^{2}}\ln\frac{2}{\alpha_{i}}\right\rceil
\label{Eq:nnfys5era}%
\end{equation}
and
\begin{equation}
\alpha_{i}=C\alpha\varepsilon K_{i}. \label{Eq:nc746ardf5}%
\end{equation}
In these formulae~$E$ and $F$ are given by~(\ref{Eq:nnc74tsfa})
and~(\ref{Eq:bc645qfsaa}), and~$C$ by~(\ref{Eq:7ag45ar}). By Hoeffding's
inequality~\cite{Hoeff} for~$N_{i}$ i.i.d.$\ $Benoulli trials applied to the
probability of~$B_{i}$ conditioned on the algorithm outcomes for rounds
$0,1,\ldots,i-1$ of the \textbf{repeat} loop, we have%
\[
\mathbb{P}(B_{i}|B_{0}\cap\cdots\cap B_{i-1})\geq1-2e^{-2N_{i}E^{2}}%
\geq1-\alpha_{i}%
\]
for each $i=0,1,\ldots\,$.

Consider the situation when the outcome of Algorithm~\ref{algQAE} is in~$B$.
Observe that%
\[
|K_{I-1}\hat{\theta}_{I-1}^{\sharp}-K_{I-1}\hat{\theta}_{I-1}^{\flat}%
|=2f(\hat{A}_{I-1},E)\leq2f(E,E)\leq2F
\]
by the definitions and properties of~$E, F, f$ (see
Section~\ref{Sect:AuxRes}) and the fact that $\sin^{2}(K_{I-1}\hat{\theta
}_{I-1}^{\flat})=\max\{\hat{A}_{I-1}-E,0\}$, $\sin^{2}(K_{I-1}\hat{\theta
}_{I-1}^{\sharp})=\min\{\hat{A}_{I-1}+E,1\}$, and $K_{I-1}\hat{\theta}%
_{I-1}^{\sharp},K_{I-1}\hat{\theta}_{I-1}^{\flat}\in\left[  m_{I-1}\frac{\pi
}{2},(m_{I-1}+1)\frac{\pi}{2}\right]  $. We have $|\hat{\theta}_{I-1}^{\sharp
}-\hat{\theta}_{I-1}^{\flat}|>2\varepsilon$ since $I$ is the smallest
$i=0,1,\ldots$ satisfying the condition $|\hat{\theta}_{i}^{\sharp}%
-\hat{\theta}_{i}^{\flat}|\leq2\varepsilon$. It follows that%
\[
K_{I-1}\leq\frac{2F}{|\hat{\theta}_{I-1}^{\sharp}-\hat{\theta}_{I-1}^{\flat}%
|}<\frac{F}{\varepsilon}.
\]
Then, for each $i=0,1,\ldots,I-1$,%
\begin{equation}
K_{i}\leq\frac{L_{I-2}}{3}\times\cdots\times\frac{L_{i}}{3}K_{i}=3^{-I+1+i}K_{I-1}%
<3^{-I+1+i}\frac{F}{\varepsilon} \label{Eq:nf756af344}%
\end{equation}
since $L_{i},\ldots,L_{I-2}\in\left\{  3,5,7\right\}  $. For $i=I$, this
inequality would read $K_{I}\leq\frac{3F}{\varepsilon}$, but this may not
hold. However, there is a higher bound for~$K_{I}$, namely,%
\begin{equation}
K_{I}<\frac{\pi}{4\varepsilon}. \label{EqL3qy7462rf}%
\end{equation}
The last inequality holds because $|\hat{\theta}_{I-1}^{\sharp}-\hat{\theta
}_{I-1}^{\sharp}|>2\varepsilon$ and $K_{I}\hat{\theta}_{I-1}^{\flat},K_{I}%
\hat{\theta}_{I-1}^{\sharp}\in\left[  m_{I}\frac{\pi}{2},\left(
m_{I}+1\right)  \frac{\pi}{2}\right]  $, which means that $|K_{I}\hat{\theta
}_{I-1}^{\sharp}-K_{I}\hat{\theta}_{I-1}^{\flat}|\leq\frac{\pi}{2}$.

It follows that%
\begin{align*}
\sum_{i=0}^{I}\alpha_{i}  &  =C\alpha\varepsilon\sum_{i=0}^{I}K_{i}%
<C\alpha\bigg(\sum_{i=0}^{I-1}3^{-I+1+i}F+\frac{\pi}{4}\bigg)=C\alpha
\bigg(\sum_{j=0}^{I-1}3^{-j}F+\frac{\pi}{4}\bigg)\\
&  <C\alpha\bigg(\sum_{j=0}^{\infty}3^{-j}F+\frac{\pi}{4}\bigg)=C\alpha
\bigg(\frac{3}{2}F+\frac{\pi}{4}\bigg)=C\alpha\frac{6F+\pi}{4}=\alpha
\end{align*}
when%
\begin{equation}
C=\frac{4}{6F+\pi}\approx0.93314. \label{Eq:7ag45ar}%
\end{equation}
Since the events $B_{0},B_{1},\ldots$ are independent of~$I$, it follows that
the conditional probability $\mathbb{P}(B|I)$ satisfies%
\[
\mathbb{P}(B|I)=\prod_{i=0}^{I}\mathbb{P}(B_{i}|B_{0}\cap\cdots\cap
B_{i-1})\geq\prod_{i=0}^{I}\left(  1-\alpha_{i}\right)  \geq1-\sum_{i=0}%
^{I}\alpha_{i}>1-\alpha,
\]
so%
\begin{equation}
\mathbb{P}(B)=\mathbb{E}\left(  \mathbb{P}(B|I)\right)  \geq1-\alpha.
\label{Eq:nbf6srdvs}%
\end{equation}

Next, we claim that%
\begin{equation}
\theta\in\lbrack\hat{\theta}_{i}^{\flat},\hat{\theta}_{i}^{\sharp}%
]\quad\text{on }B\text{ for each }i=0,1,\ldots,I. \label{Eq:jjaydrcvd9s1-3}%
\end{equation}
To verify the claim, let us assume, once again, that the outcome of
Algorithm~\ref{algQAE} is in~$B$, and proceed by induction. For $i=0$, we have
$K_{0}=0$ and $m_{0}=0$, so $\theta,\hat{\theta}_{0}^{\flat},\hat{\theta}%
_{0}^{\sharp}\in\left[  0,\frac{\pi}{2}\right]  $. With $|\hat{A}_{0}%
-A_{0}|\leq E\ $in~$B_{0}\supset B$, it means that%
\begin{align*}
\sin^{2}(\theta)=A_{0}  &  \in\lbrack\hat{A}_{0}-E,\hat{A}_{0}+E]\cap\left[
0,1\right] \\
&  =[\max\{\hat{A}_{0}-E,0\},\min\{\hat{A}_{0}+E,1\}]\\
&  =[\sin^{2}(\hat{\theta}_{0}^{\flat}),\sin^{2}(\hat{\theta}_{0}^{\sharp})],
\end{align*}
so $\theta\in\lbrack\hat{\theta}_{0}^{\flat},\hat{\theta}_{0}^{\sharp}]$. Now
suppose that $\theta\in\lbrack\hat{\theta}_{i}^{\flat},\hat{\theta}%
_{i}^{\sharp}]$ for some $i=0,1,\ldots,I-1$. According to
Algorithm~\ref{algQAE},%
\begin{equation}
\textstyle K_{i}\hat{\theta}_{i}^{\flat},K_{i}\hat{\theta}_{i}^{\sharp}%
\in\left[  m_{i}\frac{\pi}{2},\left(  m_{i}+1\right)  \frac{\pi}{2}\right]
\label{Eq:uuanera424ca}%
\end{equation}
and%
\begin{align}
\sin^{2}(K_{i}\hat{\theta}_{i}^{\flat})  &  =\max(\hat{A}_{i}%
-E,0),\label{Eq:hhatefafsr}\\
\sin^{2}(K_{i}\hat{\theta}_{i}^{\sharp})  &  =\min(\hat{A}_{i}+E,1).
\label{Eq:99ah34a3}%
\end{align}
By Lemma~\ref{Lem:bddta7d7gs}, we can find an $L_{i}\in\left\{  3,5,7\right\}
$ and an $m_{i+1}\in L_{i}m_{i}+\{0,\ldots,L_{i}-1\}$ such that%
\[
\textstyle[L_{i}K_{i}\hat{\theta}_{i}^{\flat},L_{i}K_{i}\hat{\theta}%
_{i}^{\sharp}]\subset\left[  m_{i+1}\frac{\pi}{2},\left(  m_{i+1}+1\right)
\frac{\pi}{2}\right]  .
\]
With $K_{i+1}=L_{i}K_{i}$, it follows by the induction hypothesis that%
\[
\textstyle K_{i+1}\theta\in\lbrack K_{i+1}\hat{\theta}_{i}^{\flat},K_{i+1}%
\hat{\theta}_{i}^{\sharp}]\subset\left[  m_{i+1}\frac{\pi}{2},\left(
m_{i+1}+1\right)  \frac{\pi}{2}\right]  ,
\]
and from Algorithm~\ref{algQAE} we have%
\[
\textstyle K_{i+1}\hat{\theta}_{i+1}^{\flat},K_{i+1}\hat{\theta}_{i+1}%
^{\sharp}\in\left[  m_{i+1}\frac{\pi}{2},\left(  m_{i+1}+1\right)  \frac{\pi
}{2}\right]  .
\]
We know that $|\hat{A}_{i+1}-A_{i+1}|\leq E\ $on $B$, which means that%
\begin{align*}
\sin^{2}(K_{i+1}\theta)=A_{i+1}  &  \in\lbrack\hat{A}_{i+1}-E,\hat{A}%
_{i+1}+E]\cap\left[  0,1\right] \\
&  =[\max\{\hat{A}_{i+1}-E,0\},\min\{\hat{A}_{i+1}+E,1\}]\\
&  =[\sin^{2}(K_{i+1}\hat{\theta}_{i+1}^{\flat}),\sin^{2}(K_{i+1}\hat{\theta
}_{i+}^{\sharp})].
\end{align*}
As a result,%
\[
K_{i+1}\theta\in\lbrack K_{i+1}\hat{\theta}_{i+1}^{\flat},K_{i+1}\hat{\theta
}_{i+1}^{\sharp}],
\]
so%
\[
\theta\in\lbrack\hat{\theta}_{i+1}^{\flat},\hat{\theta}_{i+1}^{\sharp}],
\]
completing the proof of the claim.

For%
\begin{equation}
\hat{a}=\textstyle\sin^{2}(\frac{1}{2}(\hat{\theta}_{I}^{\flat}+\hat{\theta
}_{I}^{\sharp})), \label{Eq:m98ab1ga}%
\end{equation}
it follows that%
\begin{align*}
|\hat{a}-a|  &  =\textstyle|\sin^{2}(\frac{1}{2}(\hat{\theta}_{I}^{\flat}%
+\hat{\theta}_{I}^{\sharp}))-\sin^{2}(\theta)|\\
&  \leq\textstyle|\frac{1}{2}(\hat{\theta}_{I}^{\flat}+\hat{\theta}%
_{I}^{\sharp})-\theta|\leq\frac{1}{2}|\hat{\theta}_{I}^{\sharp}-\hat{\theta
}_{I}^{\flat}|\leq\varepsilon\quad\text{on }B.
\end{align*}
The first inequality holds because $|\sin^{2}\alpha-\sin^{2}\beta|\leq
|\alpha-\beta|$, the second one because $\theta\in\lbrack\hat{\theta}%
_{I}^{\flat},\hat{\theta}_{I}^{\sharp}]$ on~$B$ by~(\ref{Eq:jjaydrcvd9s1-3}),
and the last one because~$I$ is the smallest $i=0,1,\ldots$ such that%
\begin{equation}
|\hat{\theta}_{i}^{\sharp}-\hat{\theta}_{i}^{\flat}|\leq2\varepsilon.
\label{Eq:g4a61dal}%
\end{equation}
Together with (\ref{Eq:nbf6srdvs}), this yields%
\[
\mathbb{P}\{|\hat{a}-a|\leq\varepsilon\}\geq\mathbb{P}(B)\geq1-\alpha,
\]
proving (\ref{Eq:ahs74ta5}).

It remains to estimate the quantum computational complexity of
Algorithm~\ref{algQAE}, understood as the number of applications~$M$ of the
unitary operator~$\mathbf{Q}$. Namely,%
\begin{align*}
M  &  =\sum_{i=0}^{I}\frac{K_{i}-1}{2}N_{i}<\frac{1}{2}\sum_{i=0}^{I}%
K_{i}N_{i}=\frac{1}{4E^{2}}\sum_{i=0}^{I}K_{i}\ln\left(  \frac{2}{\alpha_{i}%
}\right)  +\frac{1}{2}\sum_{i=0}^{I}K_{i}\\
&  =\frac{1}{4E^{2}}\sum_{i=0}^{I}K_{i}\ln\left(  \frac{2}{C\alpha\varepsilon
K_{i}}\right)  +\frac{1}{2}\sum_{i=0}^{I}K_{i}%
\end{align*}
since%
\[
N_{i}=\left\lceil \frac{1}{4E^{2}}\ln\frac{2}{\alpha_{i}}\right\rceil
\leq\frac{1}{4E^{2}}\ln\frac{2}{\alpha_{i}}+1.
\]
Observe that $x\ln\frac{c}{x}$, where~$c$ is a positive constant, is an
increasing function of $x\in(0,c/e)$. By~(\ref{Eq:nf756af344}),
(\ref{Eq:7ag45ar}), and (\ref{Eq:bc645qfsaa}),%
\[
K_{i}\leq3^{-I+1+i}\frac{F}{\varepsilon}\leq\frac{F}{\varepsilon}<\frac
{2}{C\varepsilon}\frac{1}{e}<\frac{2}{C\alpha\varepsilon}\frac{1}{e}%
\]
for each $i=0,1,\ldots,I-1$, and%
\[
K_{I}<\frac{\pi}{4\varepsilon}<\frac{2}{C\varepsilon}\frac{1}{e}<\frac
{2}{C\alpha\varepsilon}\frac{1}{e}.
\]
As a result,%
\begin{align*}
M  &  <\frac{1}{4E^{2}}%
%TCIMACRO{\TeXButton{\bigg(}{\bigg(}}%
%BeginExpansion
\bigg(%
%EndExpansion
\sum_{i=0}^{I-1}K_{i}\ln\left(  \frac{2}{C\alpha\varepsilon K_{i}}\right)
+K_{I}\ln\left(  \frac{2}{C\alpha\varepsilon K_{I}}\right)
%TCIMACRO{\TeXButton{\bigg)}{\bigg)}}%
%BeginExpansion
\bigg)%
%EndExpansion
+\frac{1}{2}\sum_{i=0}^{I-1}K_{i}+\frac{1}{2}K_{I}\\
&  <\frac{1}{4E^{2}}\frac{1}{\varepsilon}%
%TCIMACRO{\TeXButton{\bigg(}{\bigg(}}%
%BeginExpansion
\bigg(%
%EndExpansion
F\sum_{i=0}^{I-1}3^{-I+1+i}\ln\left(  \frac{2}{C\alpha3^{-I+1+i}F}\right)
+\frac{\pi}{4}\ln\left(  \frac{8}{C\alpha\pi}\right)
%TCIMACRO{\TeXButton{\bigg)}{\bigg)}}%
%BeginExpansion
\bigg)%
%EndExpansion
\\
&  \quad+\frac{1}{\varepsilon}%
%TCIMACRO{\TeXButton{\bigg(}{\bigg(}}%
%BeginExpansion
\bigg(%
%EndExpansion
\frac{1}{2}F\sum_{i=0}^{I-1}3^{-I+1+i}+\frac{\pi}{8}%
%TCIMACRO{\TeXButton{\bigg)}{\bigg)}}%
%BeginExpansion
\bigg)%
%EndExpansion
\\
&  =\frac{1}{4E^{2}}\frac{1}{\varepsilon}%
%TCIMACRO{\TeXButton{\bigg(}{\bigg(}}%
%BeginExpansion
\bigg(%
%EndExpansion
F\sum_{j=0}^{I-1}3^{-j}\ln\left(  \frac{2}{C\alpha F}\right)  +F\sum
_{j=0}^{I-1}j3^{-j}\ln\left(  3\right)  +\frac{\pi}{4}\ln\left(  \frac
{8}{C\alpha\pi}\right)
%TCIMACRO{\TeXButton{\bigg)}{\bigg)}}%
%BeginExpansion
\bigg)%
%EndExpansion
\\
&  \quad+\frac{1}{\varepsilon}%
%TCIMACRO{\TeXButton{\bigg(}{\bigg(}}%
%BeginExpansion
\bigg(%
%EndExpansion
\frac{1}{2}F\sum_{j=0}^{I-1}3^{-j}+\frac{\pi}{8}%
%TCIMACRO{\TeXButton{\bigg)}{\bigg)}}%
%BeginExpansion
\bigg)%
%EndExpansion
.
\end{align*}
Since%
\[
\sum_{j=0}^{I-1}3^{-j}<\sum_{j=0}^{\infty}3^{-j}=\frac{3}{2}\quad
\text{and}\quad\sum_{j=0}^{I-1}j3^{-j}<\sum_{j=0}^{\infty}j3^{-j}=\frac{3}%
{4},
\]
we finally get%
\begin{align*}
M  &  <\frac{1}{\varepsilon}\left(  \frac{1}{4E^{2}}\left(  \frac{3}{2}%
F\ln\left(  \frac{2}{C\alpha F}\right)  +\frac{3}{4}F\ln\left(  3\right)
+\frac{\pi}{4}\ln\left(  \frac{8}{C\alpha\pi}\right)  \right)  +\frac{3}%
{4}F+\frac{\pi}{8}\right) \\
&  \approx\frac{1}{\varepsilon}\left(  85.637-55.674\ln\left(  \alpha\right)
\right)  ,
\end{align*}
which proves that $M=O(\frac{1}{\varepsilon})$.
\end{proof}

\begin{remark}
\label{Rem:bbdta54v}\upshape It is interesting to note that the choice
of~$\alpha_{i}$ as in~(\ref{Eq:nc746ardf5}) is optimal in the following sense.
Suppose that $\alpha_{i}$ is given by a more general expression of the form%
\[
\alpha_{i}=C_{x}\alpha\left(  \varepsilon K_{i}\right)  ^{x}%
\]
for some $x>0$. Then%
\begin{align*}
\sum_{i=0}^{I}\alpha_{i}  &  =C_{x}\alpha\sum_{i=0}^{I}\left(  \varepsilon
K_{i}\right)  ^{x}\leq C_{x}\alpha%
%TCIMACRO{\TeXButton{\bigg(}{\bigg(}}%
%BeginExpansion
\bigg(%
%EndExpansion
F^{x}\sum_{i=0}^{I-1}3^{\left(  -I+1+i\right)  x}+\left(  \frac{\pi}%
{4}\right)  ^{x}%
%TCIMACRO{\TeXButton{\bigg)}{\bigg)}}%
%BeginExpansion
\bigg)%
%EndExpansion
\\
&  <C_{x}\alpha%
%TCIMACRO{\TeXButton{\bigg(}{\bigg(}}%
%BeginExpansion
\bigg(%
%EndExpansion
\frac{F^{x}}{1-3^{-x}}+\left(  \frac{\pi}{4}\right)  ^{x}%
%TCIMACRO{\TeXButton{\bigg)}{\bigg)}}%
%BeginExpansion
\bigg)%
%EndExpansion
=C_{x}\alpha%
%TCIMACRO{\TeXButton{\bigg(}{\bigg(}}%
%BeginExpansion
\bigg(%
%EndExpansion
\frac{3^{x}4^{x}F^{x}+\left(  3^{x}-1\right)  \pi^{x}}{\left(  3^{x}-1\right)
4^{x}}%
%TCIMACRO{\TeXButton{\bigg)}{\bigg)}}%
%BeginExpansion
\bigg)%
%EndExpansion
=\alpha
\end{align*}
if we take%
\[
C_{x}=\frac{\left(  3^{x}-1\right)  4^{x}}{3^{x}4^{x}F^{x}+\left(
3^{x}-1\right)  \pi^{x}}.
\]
Estimating the number of applications of~$\mathbf{Q}$ in a similar manner as
in the proof of Theorem~\ref{Thm:2and6at3l} gives%
\[
M<\frac{1}{\varepsilon}%
%TCIMACRO{\TeXButton{\bigg(}{\bigg(}}%
%BeginExpansion
\bigg(%
%EndExpansion
\frac{1}{4E^{2}}%
%TCIMACRO{\TeXButton{\bigg(}{\bigg(}}%
%BeginExpansion
\bigg(%
%EndExpansion
\frac{3}{2}F\ln%
%TCIMACRO{\TeXButton{\bigg(}{\bigg(}}%
%BeginExpansion
\bigg(%
%EndExpansion
\frac{2\sqrt{3^{x}}}{C_{x}\alpha F^{x}}%
%TCIMACRO{\TeXButton{\bigg)}{\bigg)}}%
%BeginExpansion
\bigg)%
%EndExpansion
+\frac{\pi}{4}\ln%
%TCIMACRO{\TeXButton{\bigg(}{\bigg(}}%
%BeginExpansion
\bigg(%
%EndExpansion
\frac{2^{2x+1}}{C_{x}\alpha\pi^{x}}%
%TCIMACRO{\TeXButton{\bigg)}{\bigg)}}%
%BeginExpansion
\bigg)%
%EndExpansion%
%TCIMACRO{\TeXButton{\bigg)}{\bigg)}}%
%BeginExpansion
\bigg)%
%EndExpansion
+\frac{3}{4}F+\frac{\pi}{8}%
%TCIMACRO{\TeXButton{\bigg)}{\bigg)}}%
%BeginExpansion
\bigg)%
%EndExpansion
.
\]
This upper bound for~$M$ attains its minimum value when $x=1$, that is,
$\alpha_{i}$~given by~(\ref{Eq:nc746ardf5}) with~(\ref{Eq:7ag45ar}) turns out
to be the best choice.
\end{remark}

Next, we present an accelerated version of our algorithm; see
Algorithm~\ref{algQAEacc}. We shall refer to it as the accelerated QAE (AQAE)
algorithm. Theorem~\ref{Thm:p9a3n7gd} shows the correctness of the algorithm
and provides an upper bound for quantum computational complexity.%

\begin{algorithm}
\caption{}\label{algQAEacc}
%\begin{algorithmic}[0] %without line numbers
\begin{algorithmic}[1] %with line numbers
\Require $\varepsilon>0$ and $\alpha\in(0,1)$
\State $E:=\frac{1}{2}\left(\sin^2\left(\frac{3}{7}\frac{\pi}{2}\right)
    -\sin^2\left(\frac{1}{3}\frac{\pi}{2}\right)\right)$
\State $C:=\frac{8}{3\pi}$
\State $i:=-1,\quad K_0:=1,\quad m_0:=0$
\Repeat
    \State $i:=i+1$
    \State $\alpha_i:=C\alpha\varepsilon K_i,\quad
        N_i:=\left\lceil\frac{1}{2E^2}\ln\frac{2}{\alpha_i}\right\rceil$
    \State $N:=0,\quad n:=0$
    \Repeat
        \State $N:=N+1,\quad n:=n+\texttt{\textit{QuantCirc}}( \mathbf{Q}^{(K_i-1)/2},1)$
        \State $\hat{A}_i^N:=\frac{n}{N}$
        \State $E_i^N:=\sqrt{\frac{1}{2N}\ln\frac{2}{\alpha_i}}$ if $N<N_i$,
            and $E_i^N:=E$ otherwise
        \State compute $\hat{\theta}^{\flat}_i$ and $\hat{\theta}^{\sharp}_i$ such that
            \Statex \quad\quad\quad\quad\quad
            $K_i\hat{\theta}^{\flat}_i,K_i\hat{\theta}^{\sharp}_i\in
            \big[m_i\frac{\pi}{2},(m_i+1)\frac{\pi}{2}\big]$
            \Statex \quad\quad\quad\quad\quad and\quad
            $\sin^2(K_i\hat{\theta}^{\flat}_i)=\max\{\hat{A}_i^N-E_i^N,0\}$
            \Statex \quad\quad\quad\quad\quad and\quad
            $\sin^2(K_i\hat{\theta}^{\sharp}_i)=\min\{\hat{A}_i^N+E_i^N,1\}$
    \Until{found $L_i\in\{3,5,7\}$ and $m_{i+1}\in L_im_i+\{0,\ldots,L_i-1\}$ such that}
        \Statex \quad\quad\quad
        $[L_iK_i\hat{\theta}^{\flat}_i,L_iK_i\hat{\theta}^{\sharp}_i]\subset
        \big[m_{i+1}\frac{\pi}{2},(m_{i+1}+1)\frac{\pi}{2}\big]$
%    \State $N'_i:=N,\quad \hat{A}_i:=\hat{A}_i^{N'_i},\quad E_i:=E_i^{N'_i}$
    \State $K_{i+1}:=L_iK_i$
\Until{$|\hat{\theta}^{\sharp}_i-\hat{\theta}^{\flat}_i|\le2\varepsilon$}
\State $I:=i$
\State \textbf{return} $\hat{a}:=\sin^2(\frac{1}2{}(\hat{\theta}^{\flat}_I+\hat{\theta}^{\sharp}_I))$
\end{algorithmic}
\end{algorithm}%

\begin{theorem}
\label{Thm:p9a3n7gd}\upshape Given an accuracy $\varepsilon>0$ and confidence
level $1-\alpha\in(0,1)$, Algorithm~\ref{algQAEacc} returns an estimate~$\hat
{a}$ of~$a$ such that%
\begin{equation}
\mathbb{P}\{\left\vert \hat{a}-a\right\vert \leq\varepsilon\}\geq1-\alpha.
\label{Eq:wpe8an1}%
\end{equation}
The quantum computational complexity of Algorithm~\ref{algQAEacc}, that is,
the number of applications of~$\mathbf{Q}$ also scales as $M=O(\frac
{1}{\varepsilon})$.
\end{theorem}

\begin{proof}
Fix an $a\in\left[  0,1\right]  $ and the angle $\theta\in\left[  0,\frac{\pi
}{2}\right]  $ such that%
\[
a=\sin^{2}(\theta).
\]
For each $i=0,1,\ldots~$, we put%
\[
A_{i}=\sin^{2}(K_{i}\theta).
\]
We also out%
\[
\hat{A}_{i}=\hat{A}_{i}^{N_{i}^{\prime}}%
\]
where $i$ is the iteration index of the outer \textbf{repeat} loop in
Algorithm~\ref{algQAEacc} and $N_{i}^{\prime}$ denotes the iteration inded~$N$ of
the inner \textbf{repeat} loop at which the inner loop terminates, and where%
\begin{equation}
\hat{A}_{i}^{N}=\frac{n}{N} \label{Eq:oofyrhste1}%
\end{equation}
with%
\begin{equation}
n=\textstyle%
%TCIMACRO{\TeXButton{QuantCirc}{\texttt{\textit{QuantCirc}}}}%
%BeginExpansion
\texttt{\textit{QuantCirc}}%
%EndExpansion
(\frac{K_{i}-1}{2},N). \label{Eq:bba5fes3}%
\end{equation}
It follows that~$n$ can be regarded as the number of successful outcomes
in~$N$ i.i.d.\ Bernoulli trials with probability of success~$A_{i}$. Moreover,
we put%
\[
E_{i}=E_{i}^{N_{i}^{\prime}},
\]
where%
\begin{equation}
E_{i}^{N}=\left\{
\begin{array}
[c]{cl}%
\sqrt{\frac{1}{2N}\ln\frac{2}{\alpha_{i}}} & \text{if }%
N<N_{i}\text{,}\\
E & \text{otherwise,}%
\end{array}
\right.  \label{Eq:2darrw39k}%
\end{equation}
and%
\begin{equation}
\alpha_{i}=C\alpha\varepsilon K_{i}, \label{Eq:57ahd8anvx}%
\end{equation}
with the constant~$C$ defined by~(\ref{Eq:po0in92bav}).

For the sake of this argument, suppose that the outer \textbf{repeat} loop in
Algorithm~\ref{algQAEacc} does not break when the condition $|\hat{\theta}%
_{i}^{\sharp}-\hat{\theta}_{i}^{\flat}|\leq2\varepsilon$ is satisfied, but
keeps running indefinitely. This defines $\alpha_{i},N_{i}%
,N_{i}^{\prime},\hat{A}_{i},E_{i},K_{i},\hat{\theta}_{i}^{\sharp},\hat{\theta
}_{i}^{\sharp},L_{i},m_{i}$ for each $i=0,1,\ldots$~. Then, $I$ can be taken
to be the smallest~$i=0,1,\ldots$ satisfying the condition $|\hat{\theta}%
_{i}^{\sharp}-\hat{\theta}_{i}^{\flat}|\leq2\varepsilon$.

Let%
\[
B=\textstyle\bigcap_{i=0}^{I}B_{i}\quad\text{and}\quad B^{\prime
}=\textstyle\bigcap_{i=0}^{I}B_{i}^{\prime},
\]
where, for each $i=0,1,\ldots\,$,\thinspace we denote by$~B_{i}$ the event
that the inner \textbf{repeat} loop in Algorithm~\ref{algQAEacc} terminates at
the iteration~$N=N_{i}^{\prime}=N_{i}$ and produces~$\hat{A}_{i}$ at termination
such that%
\[
|\hat{A}_{i}-A_{i}|\leq E,
\]
and by~$B_{i}^{\prime}$ we denote the event that the inner \textbf{repeat}
loop in Algorithm~\ref{algQAEacc} terminates at an iteration~$N=N_{i}^{\prime
}\leq N_{i}$ and produces~$\hat{A}_{i}$ at termination such that
\[
|\hat{A}_{i}-A_{i}|\leq E_{i},
\]
where~$E$ is given by~(\ref{Eq:nnc74tsfa}) and%
\begin{equation}
N_{i}=\left\lceil \frac{1}{2E^{2}}\ln\frac{2}{\alpha_{i}}\right\rceil
. \label{Eq:nyedftard74}%
\end{equation}
Since $E_{i}=E$ when the inner loop terminates at iteration $N=N_{i}^{\prime
}=N_{i}$, we have $B_{i}\subset B_{i}^{\prime}$ for each $i=0,1,\ldots\,$,
which implies that $B\subset B^{\prime}$.

By Hoeffding's inequality~\cite{Hoeff} for~$N_{i}$ i.i.d.$\ $Benoulli trials
applied to the probability of~$B_{i}$ conditioned on the algorithm outcomes
for rounds $0,1,\ldots,i-1$ of the outer \textbf{repeat} loop, we have%
\[
\mathbb{P}(B_{i}|B_{0}\cap\cdots\cap B_{i-1})\geq1-2e^{-2N_{i}E^{2}}%
\geq1-\alpha_{i}%
\]
for each $i=0,1,\ldots\,$.

Observe that we have $|K_{I}\hat{\theta}_{I-1}^{\sharp}-K_{I}\hat{\theta}_{I-1}%
^{\flat}|\leq\frac{\pi}{2}$ since $K_{I}\hat{\theta}_{I-1}^{\sharp},K_{I}%
\hat{\theta}_{I-1}^{\flat}\in\left[  m_{I}\frac{\pi}{2},(m_{I}+1)\frac{\pi}%
{2}\right]  $. Moreover $|\hat{\theta}_{I-1}^{\sharp}-\hat{\theta}_{I-1}%
^{\flat}|>2\varepsilon$ since $I$ is the smallest $i=0,1,\ldots$ satisfying
the condition $|\hat{\theta}_{i}^{\sharp}-\hat{\theta}_{i}^{\flat}%
|\leq2\varepsilon$. It follows that%
\[
K_{I}\leq\frac{\pi}{2|\hat{\theta}_{I-1}^{\sharp}-\hat{\theta}_{I-1}^{\flat}%
|}<\frac{\pi}{4\varepsilon}.
\]
Then, for each $i=0,1,\ldots,I$,%
\begin{equation}
K_{i}\leq\frac{L_{I-1}}{3}\times\cdots\times\frac{L_{i}}{3}K_{i}=3^{-I+i}K_{I}%
<3^{-I+i}\frac{\pi}{4\varepsilon} \label{Eq:as7bron64}%
\end{equation}
since $L_{i},\ldots,L_{I-1}\in\left\{  3,5,7\right\}  $. It follows that%
\begin{align}
\sum_{i=0}^{I}\alpha_{i}  &  =C\alpha\varepsilon\sum
_{i=0}^{I}K_{i}<C\alpha\frac{\pi}{4}\sum_{i=0}^{I}3^{-I+i}=C\alpha\frac{\pi}{4}\sum_{j=0}^{I}3^{-j}\nonumber\\
&  <C\alpha\frac{\pi}{4}\sum_{j=0}^{\infty}3^{-j}=C%
\alpha\frac{3\pi}{8}=\alpha\label{Eq:3ga63bap}%
\end{align}
when%
\begin{equation}
C=\frac{8}{3\pi}\approx0.84883. \label{Eq:po0in92bav}%
\end{equation}

Because the events $B_{0},B_{1},\ldots$ are independent of~$I$, it follows
that the conditional probability $\mathbb{P}(B|I)$ satisfies%
\[
\mathbb{P}(B|I)=\prod_{i=0}^{I}\mathbb{P}(B_{i}|B_{0}\cap\cdots\cap
B_{i-1})\geq\prod_{i=0}^{I}\left(  1-\alpha_{i}\right)  \geq1-\sum_{i=0}%
^{I}\alpha_{i}>1-\alpha,
\]
so%
\begin{equation}
\mathbb{P}(B^{\prime})\geq\mathbb{P}(B)=\mathbb{E}\left(  \mathbb{P}%
(B|I)\right)  \geq1-\alpha\label{Eq:nna64tq54po}%
\end{equation}
since $B\subset B^{\prime}$.

Next, we claim that%
\begin{equation}
\theta\in\lbrack\hat{\theta}_{i}^{\flat},\hat{\theta}_{i}^{\sharp}%
]\quad\text{on }B\text{ for each }i=0,1,\ldots,I. \label{Eq:lke5aveo}%
\end{equation}
To verify the claim, let us assume that the outcome of
Algorithm~\ref{algQAEacc} is in~$B$, and proceed by induction. For $i=0$, we
have $K_{0}=0$ and $m_{0}=0$, so $\theta,\hat{\theta}_{0}^{\flat},\hat{\theta
}_{0}^{\sharp}\in\left[  0,\frac{\pi}{2}\right]  $. Since $|\hat{A}_{0}%
-A_{0}|\leq E_{0}\ $on~$B_{0}^{\prime}$, and therefore also on $B\subset
B^{\prime}\subset B_{0}^{\prime}$, it means that%
\begin{align*}
\sin^{2}(\theta)=A_{0}  &  \in\lbrack\hat{A}_{0}-E_{0},\hat{A}_{0}+E_{0}%
]\cap\left[  0,1\right] \\
&  =[\max\{\hat{A}_{0}-E_{0},0\},\min\{\hat{A}_{0}+E_{0},1\}]\\
&  =[\sin^{2}(\hat{\theta}_{0}^{\flat}),\sin^{2}(\hat{\theta}_{0}^{\sharp})]
\end{align*}
when the inner loop terminates, so $\theta\in\lbrack\hat{\theta}_{0}^{\flat
},\hat{\theta}_{0}^{\sharp}]$. Now suppose that $\theta\in\lbrack\hat{\theta
}_{i}^{\flat},\hat{\theta}_{i}^{\sharp}]$ for some $i=0,1,\ldots,I-1$.
According to Algorithm~\ref{algQAEacc},%
\begin{equation}
\textstyle K_{i}\hat{\theta}_{i}^{\flat},K_{i}\hat{\theta}_{i}^{\sharp}%
\in\left[  m_{i}\frac{\pi}{2},\left(  m_{i}+1\right)  \frac{\pi}{2}\right]
\label{Eq:65af1vzjl}%
\end{equation}
and
\begin{align}
\sin^{2}(K_{i}\hat{\theta}_{i}^{\flat})  &  =\max(\hat{A}_{i}-E_{i}%
,0),\label{Eq:884ha846}\\
\sin^{2}(K_{i}\hat{\theta}_{i}^{\sharp})  &  =\min(\hat{A}_{i}+E_{i},1)
\label{Eq:3ndta43vd}%
\end{align}
when the inner loop terminates. If no $L_{i}\in\left\{  3,5,7\right\}  $ and
$m_{i+1}\in L_{i}m_{i}+\{0,\ldots,L_{i}-1\}$ such that%
\[
\textstyle[L_{i}K_{i}\hat{\theta}_{i}^{\flat},L_{i}K_{i}\hat{\theta}%
_{i}^{\sharp}]\subset\left[  m_{i+1}\frac{\pi}{2},\left(  m_{i+1}+1\right)
\frac{\pi}{2}\right]
\]
has been found for any of the iterations of the inner loop in
Algorithm~\ref{algQAEacc} with index~$N<N_{i}$, then
Lemma~\ref{Lem:bddta7d7gs} ensures that such an~$L_{i}$ and~$m_{i+1}$ can be
found for the iteration of the inner loop with index $N=N_{i}$. This is
because, according to~(\ref{Eq:2darrw39k}), $E_{i}=E_{i}^{N_{i}}=E$ when
$N=N_{i}$.\ With $K_{i+1}=L_{i}K_{i}$, it follows by the induction hypothesis
$\theta\in\lbrack\hat{\theta}_{i}^{\flat},\hat{\theta}_{i}^{\sharp}]$ that%
\begin{equation}
\textstyle K_{i+1}\theta\in\lbrack K_{i+1}\hat{\theta}_{i}^{\flat},K_{i+1}%
\hat{\theta}_{i}^{\sharp}]\subset\left[  m_{i+1}\frac{\pi}{2},\left(
m_{i+1}+1\right)  \frac{\pi}{2}\right]  . \label{Eq:hhf645n7}%
\end{equation}
From Algorithm~\ref{algQAEacc} we have%
\begin{equation}
\textstyle K_{i+1}\hat{\theta}_{i+1}^{\flat},K_{i+1}\hat{\theta}_{i+1}%
^{\sharp}\in\left[  m_{i+1}\frac{\pi}{2},\left(  m_{i+1}+1\right)  \frac{\pi
}{2}\right]  . \label{Eq:a5ssfa53nl}%
\end{equation}
We know that $|\hat{A}_{i+1}-A_{i+1}|\leq E_{i+1}\ $on~$B_{i+1}^{\prime}$, hence
on $B\subset B_{i+1}\subset B_{i+1}^{\prime}$, which means that%
\begin{align*}
\sin^{2}(K_{i+1}\theta)=A_{i+1}  &  \in\lbrack\hat{A}_{i+1}-E_{i+1},\hat
{A}_{i+1}+E_{i+1}]\cap\left[  0,1\right] \\
&  =[\max\{\hat{A}_{i+1}-E_{i+1},0\},\min\{\hat{A}_{i+1}+E_{i+1},1\}]\\
&  =[\sin^{2}(K_{i+1}\hat{\theta}_{i+1}^{\flat}),\sin^{2}(K_{i+1}\hat{\theta
}_{i+1}^{\sharp})]
\end{align*}
when the inner loop terminates. Together with~(\ref{Eq:hhf645n7})
and~(\ref{Eq:a5ssfa53nl}), this implies that%
\[
K_{i+1}\theta\in\lbrack K_{i+1}\hat{\theta}_{i+1}^{\flat},K_{i+1}\hat{\theta
}_{i+1}^{\sharp}],
\]
so%
\[
\theta\in\lbrack\hat{\theta}_{i+1}^{\flat},\hat{\theta}_{i+1}^{\sharp}],
\]
completing the proof of the claim.

With%
\begin{equation}
\hat{a}=\textstyle\sin^{2}(\frac{1}{2}(\hat{\theta}_{I}^{\flat}+\hat{\theta
}_{I}^{\sharp})), \label{Eq:ooijhavv23s}%
\end{equation}
it follows that%
\begin{align*}
|\hat{a}-a|  &  =\textstyle|\sin^{2}(\frac{1}{2}(\hat{\theta}_{I}^{\flat}%
+\hat{\theta}_{I}^{\sharp}))-\sin^{2}(\theta)|\\
&  \leq\textstyle|\frac{1}{2}(\hat{\theta}_{I}^{\flat}+\hat{\theta}%
_{I}^{\sharp})-\theta|\leq\frac{1}{2}|\hat{\theta}_{I}^{\sharp}-\hat{\theta
}_{I}^{\flat}|\leq\varepsilon\quad\text{on }B.
\end{align*}
The first inequality holds because $|\sin^{2}\alpha-\sin^{2}\beta|\leq
|\alpha-\beta|$, the second one because $\theta\in\lbrack\hat{\theta}%
_{I}^{\flat},\hat{\theta}_{I}^{\sharp}]$ on~$B$ by~(\ref{Eq:lke5aveo}), and
the last one because~$I$ is the smallest $i=0,1,2,\ldots$ such that%
\begin{equation}
|\hat{\theta}_{i}^{\sharp}-\hat{\theta}_{i}^{\flat}|\leq2\varepsilon.
\label{Eq:ffle5a43b}%
\end{equation}
Together with (\ref{Eq:nna64tq54po}), this yields%
\[
\mathbb{P}\{|\hat{a}-a|\leq\varepsilon\}\geq\mathbb{P}(B)\geq1-\alpha,
\]
proving (\ref{Eq:wpe8an1}).

It remains to estimate the quantum computational complexity of
Algorithm~\ref{algQAEacc}, that is, the number~$M$ of applications of the unitary
operator~$\mathbf{Q}$. Namely,%
\begin{align*}
M  &  =\sum_{i=0}^{I}\frac{K_{i}-1}{2}N_{i}^{\prime}<\frac{1}{2}\sum_{i=0}%
^{I}K_{i}N_{i}<\frac{1}{4E^{2}}\sum_{i=0}^{I}K_{i}\ln\left(  \frac{2}%
{\alpha_{i}}\right)  +\frac{1}{2}\sum_{i=0}^{I}K_{i}\\
&  =\frac{1}{4E^{2}}\sum_{i=0}^{I}K_{i}\ln\left(  \frac{2}{C%
\alpha\varepsilon K_{i}}\right)  +\frac{1}{2}\sum_{i=0}^{I}K_{i}%
\end{align*}
since%
\[
N_{i}^{\prime}\leq N_{i}=\left\lceil \frac{1}{2E^{2}}\ln\frac{2}{\alpha
_{i}}\right\rceil <\frac{1}{2E^{2}}\ln\frac{2}{\alpha_{i}%
}+1.
\]
Observe that $x\ln\frac{c}{x}$, where~$c$ is a positive constant, is an
increasing function of $x\in(0,c/e)$. By~(\ref{Eq:as7bron64}),%
\[
K_{i}<3^{-I+i}\frac{\pi}{4\varepsilon}\leq\frac{\pi}{4\varepsilon}<\frac
{2}{C\varepsilon}\frac{1}{e}<\frac{2}{C\alpha\varepsilon
}\frac{1}{e}%
\]
for each $i=0,1,\ldots,I$. As a result,%
\begin{align*}
M  &  <\frac{1}{4E^{2}}\frac{\pi}{4\varepsilon}\sum_{i=0}^{I}3^{-I+i}%
\ln\left(  \frac{8}{C^\alpha3^{-I+i}\pi}\right)  +\frac{\pi
}{8\varepsilon}\sum_{i=0}^{I}3^{-I+i}\\
&  =\frac{1}{4E^{2}}\frac{\pi}{4\varepsilon}\sum_{j=0}^{I}3^{-j}\ln\left(
\frac{8}{C\alpha3^{-j}\pi}\right)  +\frac{\pi}{8\varepsilon}%
\sum_{j=0}^{I}3^{-j}\\
&  <\frac{1}{4E^{2}}\frac{\pi}{4\varepsilon}\sum_{j=0}^{\infty}3^{-j}%
\ln\left(  \frac{8}{C\alpha3^{-j}\pi}\right)  +\frac{\pi
}{8\varepsilon}\sum_{j=0}^{\infty}3^{-j}\\
&  =\frac{1}{4E^{2}}\frac{\pi}{4\varepsilon}\left(  \sum_{j=0}^{\infty}%
3^{-j}\ln\left(  \frac{8}{C\alpha\pi}\right)  +\sum_{j=0}^{\infty
}3^{-j}j\ln\left(  3\right)  \right)  +\frac{\pi}{8\varepsilon}\sum
_{j=0}^{\infty}3^{-j}%
\end{align*}
Since%
\[
\sum_{j=0}^{I-1}3^{-j}<\sum_{j=0}^{\infty}3^{-j}=\frac{3}{2}\quad
\text{and}\quad\sum_{j=0}^{I-1}j3^{-j}<\sum_{j=0}^{\infty}j3^{-j}=\frac{3}%
{4},
\]
we finally get%
\[
M<\frac{1}{\varepsilon}\left(  \frac{1}{4E^{2}}\frac{\pi}{4}\left(  \frac
{3}{2}\ln\left(  \frac{8}{C\alpha\pi}\right)  +\frac{3}{4}\ln\left(
3\right)  \right)  +\frac{3\pi}{16}\right)  =O\left(  \frac{1}{\varepsilon
}\right)  .
\]%
\end{proof}

\begin{remark}
\label{Rem:awd6a54a}\upshape Just like observed in Remark~\ref{Rem:bbdta54v},
in the case of Algorithm~\ref{algQAEacc} the choice of~$\alpha_{i}$
given by~(\ref{Eq:57ahd8anvx}) is optimal in the following sense. Suppose
that~$\alpha_{i}$ is given by an expression of the form%
\[
\alpha_{i}=C_{x}\alpha(\varepsilon K_{i})^{x}%
\]
for some $x>0$. Then%
\begin{align*}
\sum_{i=0}^{I}\alpha_{i}  &  =C_{x}\alpha\sum_{i=0}^{I}\left(
\varepsilon K_{i}\right)  ^{x}<C_{x}\alpha\sum_{i=0}^{I}\left(
3^{-I+i}\frac{\pi}{4}\right)  ^{x}<C_{x}\alpha\left(  \frac{\pi}%
{4}\right)  ^{x}\sum_{j=0}^{\infty}3^{-jx}\\
&  =C_{x}\alpha\left(  \frac{\pi}{4}\right)  ^{x}\frac{1}{1-3^{-x}%
}=C_{x}\alpha\frac{\pi^{x}}{4^{x}\left(  1-3^{-x}\right)  }=\alpha
\end{align*}
if we take%
\[
C_{x}=\frac{4^{x}\left(  1-3^{-x}\right)  }{\pi^{x}}.
\]
Then the number of applications of~$\mathbf{Q}$ can be estimated as%
\[
M<\frac{1}{\varepsilon}\left(  \frac{1}{4E^{2}}\frac{\pi}{4}\frac{3}{2}%
\ln\left(  \frac{2\sqrt{3^{x}}}{\left(  1-3^{-x}\right)  \alpha}\right)
+\frac{3\pi}{16}\right)  .
\]
This upper bound for~$M$ attains its minimum when $x=1$, meaning
that~$\alpha_{i}$ given by~(\ref{Eq:57ahd8anvx})
with~(\ref{Eq:po0in92bav}) is the best choice.
\end{remark}

\section{Auxiliary results and notation\label{Sect:AuxRes}}

For any $a\in\left[  0,1\right]  $ and $L\in\left\{  3,5,7,\ldots\right\}  $,
we put%
\[
e_{L}(a)=\min_{l=1,\ldots,L-1}\left\vert a-\sin^{2}\left(  \frac{l}{L}%
\frac{\pi}{2}\right)  \right\vert .
\]
Moreover, let%
\begin{equation}
e(a)=\max_{L=3,5,7,\ldots}e_{L}(a). \label{eq:hdf76a55a}%
\end{equation}
The function $e(a)$ is shown in Figure~\ref{Fig:hfdtsds}, along with
$e_{3}(a)$, $e_{5}(a)$, and $e_{7}(a)$. \begin{figure}[th]
\begin{center}
\includegraphics[
width=10cm,
]{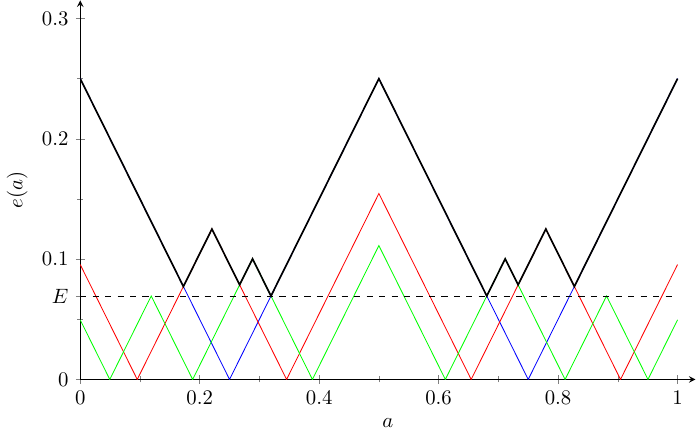}
\end{center}
\caption{The functions $e(a)$ (black) and $e_{3}(a)$, $e_{5}(a),$ $e_{7}(a)$
(blue, red, green). The dashed line indicates the minimum value~$E$
of the function~$e(a)$.}%
\label{Fig:hfdtsds}%
\end{figure}

\begin{lemma}
\label{Lem:bddta7d7gs}\upshape Fix any $a\in\left[  0,1\right]  $,
$\varepsilon>0$, $n\in\left\{  0,1,2,\ldots\right\}  $, and take the angles
$\theta_{a,\varepsilon,n}^{\flat},\theta_{a,\varepsilon,n}^{\sharp}\in\left[
n\frac{\pi}{2},\left(  n+1\right)  \frac{\pi}{2}\right]  $ such that%
\[
\sin^{2}(\theta_{a,\varepsilon,n}^{\flat})=\max\left(  a-\varepsilon,0\right)
,\quad\sin^{2}(\theta_{a,\varepsilon,n}^{\sharp})=\min\left(  a+\varepsilon
,1\right)  .
\]
Then%
\begin{equation}
0<\varepsilon\leq e(a) \label{eq:ndf7s5a4}%
\end{equation}
if an only if there is an $L\in\left\{  3,5,7,\ldots\right\}  $ such that%
\begin{equation}
\lbrack L\theta_{a,\varepsilon,n}^{\flat},L\theta_{a,\varepsilon,n}^{\sharp
}]\subset\left[  m\frac{\pi}{2},\left(  m+1\right)  \frac{\pi}{2}\right]
\quad\text{for some }m\in Ln+\left\{  0,1,\ldots,L-1\right\}  .
\label{eq:nfbf6srs}%
\end{equation}

\end{lemma}

\begin{proof}
Condition (\ref{eq:ndf7s5a4}) is satisfied if and only if there is an
$L\in\left\{  3,5,7,\ldots\right\}  $ such that%
\begin{equation}
0<\varepsilon\leq e_{L}(a). \label{eq:hdftsfsca3}%
\end{equation}
Now observe that condition~(\ref{eq:hdftsfsca3}) holds whenever%
\begin{align*}
a  &  \in\left[  0,\sin^{2}\left(  \frac{1}{L}\frac{\pi}{2}\right)
-\varepsilon\right] \\
\text{or\quad}a  &  \in\left[  \sin^{2}\left(  \frac{1}{L}\frac{\pi}%
{2}\right)  +\varepsilon,\sin^{2}\left(  \frac{2}{L}\frac{\pi}{2}\right)
-\varepsilon\right] \\
\vdots & \\
\text{or\quad}a  &  \in\left[  \sin^{2}\left(  \frac{L-2}{L}\frac{\pi}%
{2}\right)  +\varepsilon,\sin^{2}\left(  \frac{L-1}{L}\frac{\pi}{2}\right)
-\varepsilon\right] \\
\text{or\quad}a  &  \in\left[  \sin^{2}\left(  \frac{L-1}{L}\frac{\pi}%
{2}\right)  +\varepsilon,1\right]  ,
\end{align*}
and so does condition~(\ref{eq:nfbf6srs}). This proves the lemma.
\end{proof}

\begin{remark}
\label{Rem:jjd7a5a}\upshape In fact, condition (\ref{eq:ndf7s5a4}) is
satisfied if and only if there is an $L\in\left\{  3,5,7\right\}  $ such
that\ (\ref{eq:nfbf6srs})~holds. This is so because, for each $a\in\left[
0,1\right]  $, the maximum in~(\ref{eq:hdf76a55a}) is attained at an
$L\in\left\{  3,5,7\right\}  $ (which depends on~$a$).
\end{remark}

The largest possible value of~$\varepsilon$ such that (\ref{eq:ndf7s5a4})
holds for all $a\in\left[  0,1\right]  $ is denoted by~$E$\thinspace. This
value, indicated by the dashed horizontal line in Figure~\ref{Fig:hfdtsds}, is%
\begin{equation}
E=\min_{a\in\lbrack0,1]}e(a)=\frac{1}{2}\left(  \sin^{2}\left(  \frac{3}%
{7}\frac{\pi}{2}\right)  -\sin^{2}\left(  \frac{1}{3}\frac{\pi}{2}\right)
\right)  \approx0.06937, \label{Eq:nnc74tsfa}%
\end{equation}
with the minimum attained at $a=\frac{1}{2}\left(  \sin^{2}\left(  \frac{3}%
{7}\frac{\pi}{2}\right)  +\sin^{2}\left(  \frac{1}{3}\frac{\pi}{2}\right)
\right)  $. (Note that the minimum is also attained at $a=\frac{1}{2}\left(
\sin^{2}\left(  \frac{4}{7}\frac{\pi}{2}\right)  +\sin^{2}\left(  \frac{2}%
{3}\frac{\pi}{2}\right)  \right)  $.)\ By Lemma~\ref{Lem:bddta7d7gs} and
Remark~\ref{Rem:jjd7a5a}, for every $a\in\left[  0,1\right]  $ and
$\varepsilon\in(0,E]$, there is an $L\in\left\{  3,5,7\right\}  $ such that
(\ref{eq:nfbf6srs}) holds.

We take%
\begin{align*}
f(a,\varepsilon)  &  =\frac{1}{2}|\theta_{a,\varepsilon,n}^{\sharp}%
-\theta_{a,\varepsilon,n}^{\flat}|\\
&  =\frac{1}{2}\left(  \arcsin\sqrt{\min\left(  a+\varepsilon,1\right)
}-\arcsin\sqrt{\max\left(  a-\varepsilon,0\right)  }\right)
\end{align*}
(where $\theta_{a,\varepsilon,n}^{\sharp},\theta_{a,\varepsilon,n}^{\flat}$
are defined in Lemma~\ref{Lem:bddta7d7gs}) as the accuracy of estimating the
angle $\theta$ such that $\sin^{2}\theta=a$, expressed as a function of $a$
and the accuracy~$\varepsilon>0$ of estimating~$a\in\left[  0,1\right]  $.
Note that $f(a,\varepsilon)$ does not depend on~$n$. The graph of $f(a,E)$ is
shown in Figure~\ref{Fig:nnfys65}. We also put%
\begin{equation}
F=\max_{a\in\left[  0,1\right]  }f(a,E), \label{Eq:ftsfa9am1}%
\end{equation}
indicated by the dashed horizontal line in Figure~\ref{Fig:nnfys65}. The maximum is
attained at $a=E$ (and also at $a=1-E$), so%
\begin{align}
F  &  =f(E,E)=\frac{1}{2}\arcsin\sqrt{2E}\nonumber\\
&  =\frac{1}{2}\arcsin\sqrt{\sin^{2}\left(  \frac{3}{7}\frac{\pi}{2}\right)
-\sin^{2}\left(  \frac{1}{3}\frac{\pi}{2}\right)  }\approx0.19084,
\label{Eq:bc645qfsaa}%
\end{align}
This is the maximum accuracy of estimating~$\theta$ when the accuracy of
estimating\ $a=\sin^{2}\theta$ is$~E$. \begin{figure}[th]
\begin{center}
\includegraphics[
width=10cm,
]{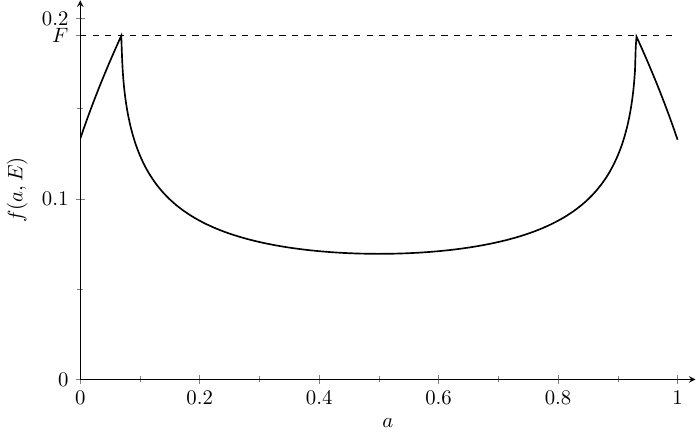}
\end{center}
\caption{The function $f(a,E)$. The dashed line indicates the maximum
value~$F$.}%
\label{Fig:nnfys65}%
\end{figure}

\section{Algorithm performance\label{Sect:hhf6sgx}}

The quantum computational complexity of the algorithm, which we denote by~$M$,
is understood as the number of applications of the operator~$\mathbf{Q}$. We
have shown that $M=O(\frac{1}{\varepsilon})$ for the QAE algorithms put
forward in the present paper. In fact, the constants in our $O(\frac
{1}{\varepsilon})$ bound for $M$ are low enough to outperform those in the
earlier papers, in particular~\cite{Gri2021} and~\cite{Fuk2023}. This can be
seen in Figure~\ref{Fig:anametf4d} by comparing the line labelled as
$\mathrm{AQAE}_{M}$ computed from~(\ref{Eq:bbc65a43}) in the present paper with the
lines $\mathrm{IQAE}_{M}$\ and $\mathrm{MIQAE}_{M}$ representing the bound%
\[
M<\frac{50}{\varepsilon}\ln\left(  \frac{1}{\alpha}\log_{2}\left(  \frac{\pi
}{4\varepsilon}\right)  \right)
\]
for the IQAE algorithm of Grinko \emph{et al.}~\cite{Gri2021}, formula~(8),
and the bound%
\[
M<\frac{62}{\varepsilon}\ln\left(  \frac{6}{\alpha}\right)
\]
for the MIQAE algorithm of Fukuzawa \emph{et al.}~\cite{Fuk2023},
formula~(3.32). Additionally, Figure~\ref{Fig:anametf4d} shows the bound
$\mathrm{AQAE}_{\mathbb{E}(M)}$ on the expected value of~$M$ obtained in
Section~\ref{Sect:f6a9m3} of the present paper.

\begin{landscape}
\begin{figure}[th]
\begin{center}
\includegraphics[
width=20cm,
]{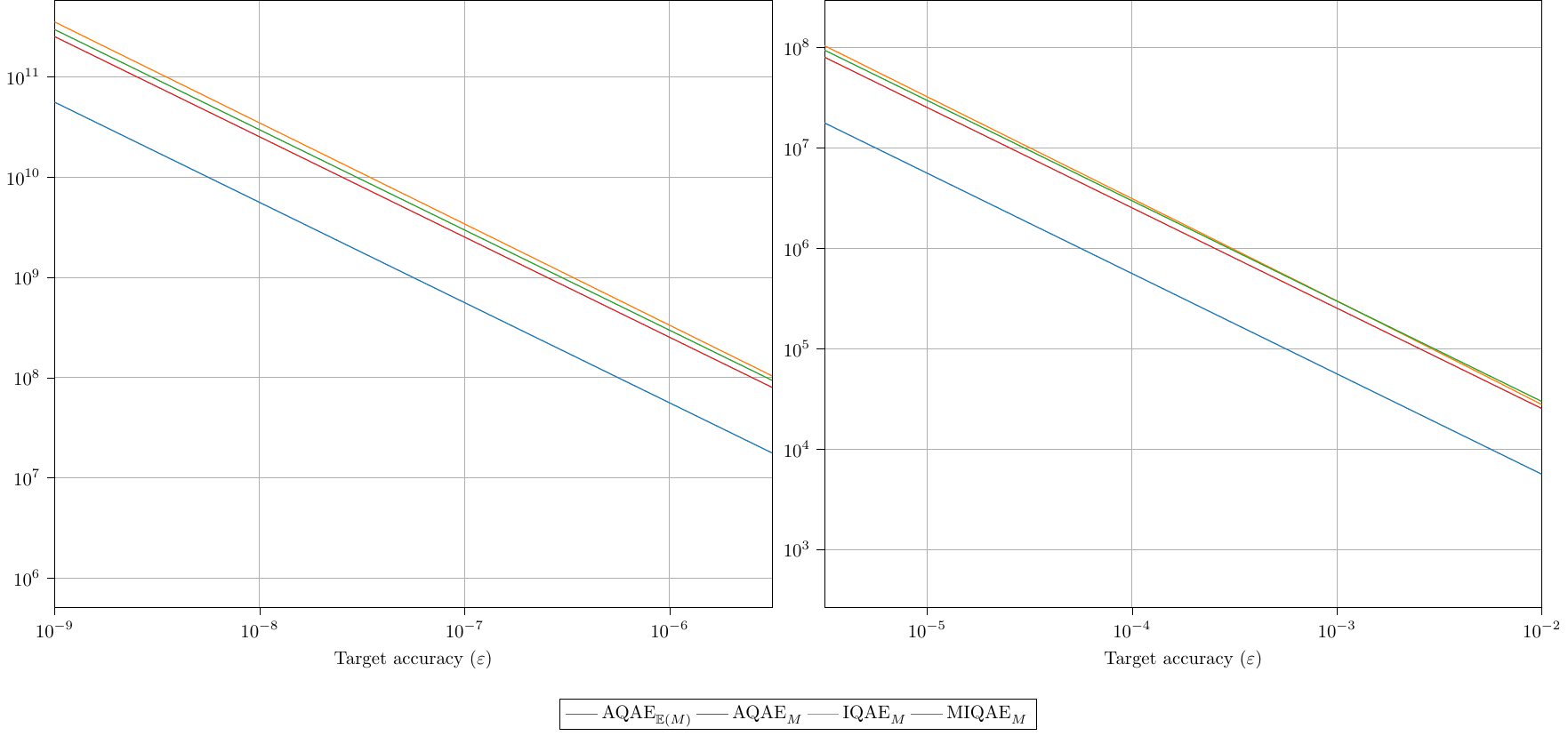}
\end{center}
\caption{Bounds on quantum computational complexity~$M$ depending on
the target accuracy~$\varepsilon$ as compared to the bounds for the IAQE
algorithm~\cite{Gri2021} and the MIQAE algorithm~\cite{Fuk2023} in the case
when $a=0.5$ and $\alpha=0.05$. A bound on the expectation $\mathbb{E}(M)$ for
the AQAE algorithm is also shown (see Section~\ref{Sect:f6a9m3}).}%
\label{Fig:anametf4d}%
\end{figure}
\end{landscape}

\begin{landscape}
\begin{figure}[th]
\begin{center}
\includegraphics[
width=20cm,
]{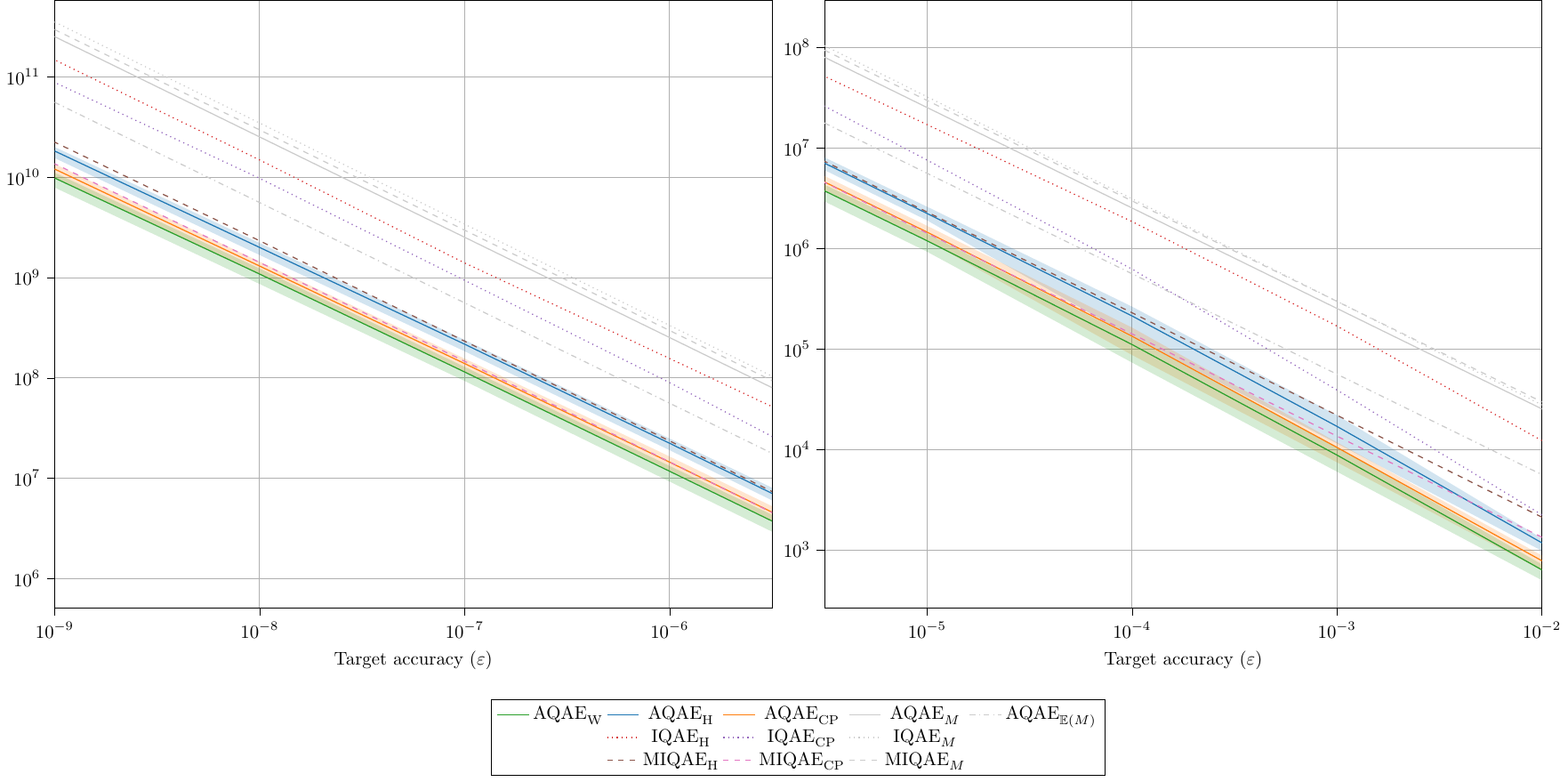}
\end{center}
\caption{Quantum computational complexity~$M$ depending on
the target accuracy~$\varepsilon$ as compared to the results from the IAQE
algorithm~\cite{IQAE} and MIQAE algorithm~\cite{MIQAE} in the case when
$a=0.5$ and $\alpha=0.05$. The values of~$M$ are averages over 2000 runs of the algorithms. The shaded areas indicate the 25\%--75\%
interquartile ranges. The gray lines show the theoretical bounds from
Figure~\ref{Fig:anametf4d}.}%
\label{Fig:aur1ba72}%
\end{figure}
\end{landscape}

Numerical experiments to assess the performance of Algorithm~\ref{algQAEacc}
are presented in Figure~\ref{Fig:aur1ba72} and compared with those from the
Qiskit implementation \cite{IQAE} of the IQAE algorithm and the
implementation~\cite{MIQAE} of the MIQAE algorithm.

For an implementation of the AQAE algorithm, see~\cite{AQAE}. In
addition to the Hoeffding confidence interval
\[
\lbrack\max(\hat{A}_{i}^{N}-E_{i}^{N},0),\min(\hat{A}_{i}^{N}-E_{i}^{N},0)]
\]
with~$E_{i}$ given by~(\ref{Eq:2darrw39k}) and~$\hat{A}_{i}^{N}$
by~(\ref{Eq:oofyrhste1}), we also implement Algorithm~\ref{algQAEacc} with the
Clopper--Pearson confidence interval
\[
\textstyle\left[  B\left(  \frac{\alpha_{i}}{2};n,N-n+1\right)
,B\left(  1-\frac{\alpha_{i}}{2};n+1,N-n\right)  \right]  ,
\]
where $B\left(  p;a,b\right)  $ is the $p$-th quantile of the beta
distribution with shape parameters~$a$ and~$b$. In Figure~\ref{Fig:aur1ba72}
the lines for the AQAE, IQAE and MIQAE algorithms with Hoeffding and
Clopper--Pearson confidence intervals are indicated by subscripts $H$
and~$CP$, respectively. Because $M$, the number of applications~of~$\mathbf{Q}%
$, is random, we show the averages of~$M$ over 2000 runs of the algorithms. To
facilitate comparison, we show the 25\%--75\% interquartile ranges (rendered
as shaded bands) for~$M$ from our AQAE algorithm. This shows superior
performance of Algorithm~\ref{algQAEacc} compared to~\cite{IQAE}
and~\cite{MIQAE}, and therefore compared to all the earlier algorithms.

The Hoeffding and Clopper--Pearson confidence intervals are exact in the sense
of guaranteeing the given confidence level for all possible values of the
binomial distribution probability parameter being estimated. In doing so, they
are also conservative. The performance of the algorithm can be improved
further by allowing approximate confidence intervals, which are narrower and
meet the prescribed confidence level approximately. This is possible because
the sum $\sum_{i=0}^{I}\alpha_{i}$, for which~$\alpha$ is an upper
bound as shown in~(\ref{Eq:3ga63bap}), in fact turns out significantly lower
than~$\alpha$ in numerical experiments. Here we implement Wilson's score
interval, the approximate confidence interval first proposed by
Wilson~\cite{Wil1927} in 1927; see also~\cite{AgCo1998}:%
\[
\textstyle\left[  \frac{\hat{A}_{i}^{N}+\frac{z_{\alpha_{i}/2}^{2}%
}{2N}-z_{\alpha_{i}/2}\sqrt{\frac{\hat{A}_{i}^{N}\left(  1-\hat
{A}_{i}^{N}\right)  }{N}+\frac{z_{\alpha_{i}/2}^{2}}{4N^{2}}}%
}{1+\frac{z_{\alpha_{i}/2}^{2}}{N}},\frac{\hat{A}_{i}^{N}%
+\frac{z_{\alpha_{i}/2}^{2}}{2N}+z_{\alpha_{i}/2}\sqrt
{\frac{\hat{A}_{i}^{N}\left(  1-\hat{A}_{i}^{N}\right)  }{N}+\frac
{z_{\alpha_{i}/2}^{2}}{4N^{2}}}}{1+\frac{z_{\alpha_{i}%
/2}^{2}}{N}}\right]  ,
\]
where $z_{c}$ denotes the $1-c$ quantile of the standard normal distribution.
Wilson's score interval does indeed improve the performance of the algorithm
still further, as can be seen in Figure~\ref{Fig:aur1ba72}, line
$\mathrm{AQAE}_{W}$.

\subsection{Computer-assisted bound for the expectation of~$M$%
\label{Sect:f6a9m3}}

The theoretical bound $\mathrm{AQAE}_{M}$ for the number of applications~$M$
of~$\mathbf{Q}$ is quite high compared to the results $\mathrm{AQAE}_{H}$ (and
also $\mathrm{AQAE}_{CP}$ and $\mathrm{AQAE}_{W}$) of numerical experiments in
Figure~\ref{Fig:aur1ba72}. Since~$M$ is random in Algorithm~\ref{algQAEacc}, a
bound for the expected value of~$M$ would be more relevant. We achieve such a
bound by first evaluating the expectation of the number of iterations~$N_{i}%
^{\prime}$ of the inner loop for each round~$i$ of the outer loop in
Algorithm~\ref{algQAEacc}, conditioned on the algorithm outcomes for rounds
$0,1,\ldots,i-1$ and on~$I$. This involves computer-assisted computations of
the binomial distribution probabilities of certain events.

The value of~$N_{i}^{\prime}$ is obtained by running the inner loop in
Algorithm~\ref{algQAEacc}. From~$\alpha_{i}$ we compute~$E_{i}^{N}$
using~(\ref{Eq:2darrw39k}). According to Lemma~\ref{Lem:bddta7d7gs}, the inner
loop stops at an iteration~$N=N_{i}^{\prime}$ such that $E_{i}^{N}\leq
e(\hat{A}_{i}^{N})$. This allows us to view~$N_{i}^{\prime}$ as a random
variable with values in $\{1,\ldots,N_{i}\}$ depending on~$\alpha_{i}$, and to view the conditional expectation of~$N_{i}^{\prime}$ as a function
of~$\alpha_{i}$.

We compute the probability $\mathbb{P}_{i}(E_{i}^{N}\leq e(\hat{A}_{i}^{N}))$
for each~$N$ in $\{1,\ldots,N_{i}\}$, conditioned on the algorithm outcomes
for rounds $0,1,\ldots,i-1$ and on~$I$. Then we use these conditional
probabilities to evaluate a bound for the expectation~$\mathbb{E}_{i}%
(N_{i}^{\prime})$, also conditioned on the algorithm outcomes for rounds
$0,1,\ldots,i-1$ and on~$I$. Observe that $E_{i}^{N}\leq e(\hat{A}_{i}^{N})$
whenever~$\hat{A}_{i}^{N}$ belongs to the union of the intervals depicted in
Figure~\ref{Fig:la8n7a54v}. This union of intervals will be denoted
by~$C_{i}^{N}$. We have $\hat{A}_{i}^{N}\in C_{i}^{N}$ iff $n\in NC_{i}^{N}$,
where~$n$ follows the binomial distribution $B(A_{i},N)$, that is, it is the
number of successful outcomes in $N$ independent Bernoulli trials with
probability of success~$A_{i}$. This allows the conditional probability of
$E_{i}^{N}\leq e(\hat{A}_{i}^{N})$ to be expressed as%
\[
\mathbb{P}_{i}(\hat{A}_{i}^{N}\in C_{i}^{N})=\mathbb{P}_{i}(E_{i}^{N}\leq
e(\hat{A}_{i}^{N}))=\sum_{n\in NC_{i}}\binom{N}{n}A_{i}^{n}(1-A_{i})^{N-n}.
\]
\begin{figure}[th]
\begin{center}
\includegraphics[
width=10cm,
]{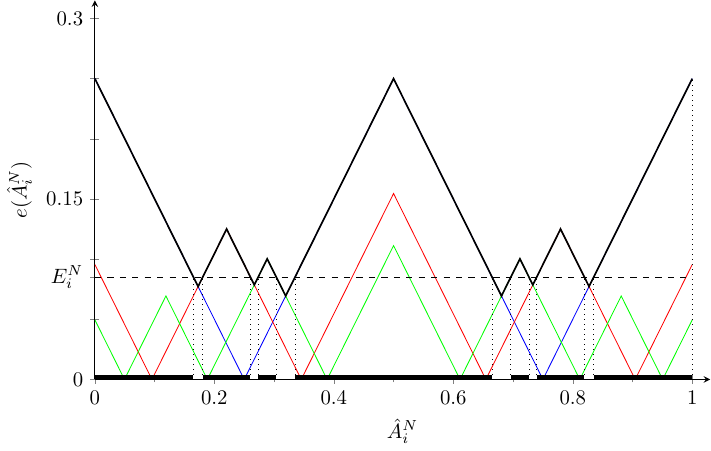}
\end{center}
\caption{The set $C_{i}^{N}$ of values of $\hat{A}_{i}^{N}$ such that
$E_{i}^{N}\leq e(\hat{A}_{i}^{N})$ is the union of the bold intervals on the
horizontal axis.}%
\label{Fig:la8n7a54v}%
\end{figure}

Then, a bound for the conditional expectation~$\mathbb{E}_{i}(N_{i}^{\prime})$
of the number of shots~$N_{i}^{\prime}$ in the~$i$th round of the outer loop
can be obtained as follows. For any $J\in\left\{  1,\ldots,N_{i}\right\}  $,
we have%
\begin{align*}
\mathbb{E}_{i}(N_{i}^{\prime}) &  =\sum_{N=1}^{N_{i}}N\mathbb{P}_{i}\left(
N=N_{i}^{\prime}\right)  =\sum_{N=1}^{J}N\mathbb{P}_{i}\left(  N=N_{i}%
^{\prime}\right)  +\sum_{N=J+1}^{N_{i}}N\mathbb{P}_{i}\left(  N=N_{i}^{\prime
}\right)  \\
&  \leq J\mathbb{P}_{i}\left(  J\geq N_{i}^{\prime}\right)  +N_{i}%
\mathbb{P}_{i}\left(  J<N_{i}^{\prime}\right)  \leq J+(N_{i}-J)\mathbb{P}%
_{i}\left(  J<N_{i}^{\prime}\right)  .
\end{align*}
Since $J<N_{i}$ implies that $E_{i}^{J}>e(\hat{A}_{i}^{J})$, it follows that%
\[
\mathbb{E}_{i}(N_{i}^{\prime})\leq J+(N_{i}-J)\mathbb{P}_{i}(E_{i}^{J}%
>e(\hat{A}_{i}^{J}))=J+(N_{i}-J)\left(  1-\mathbb{P}_{i}(\hat{A}_{i}^{J}\in
C_{i}^{J})\right)  .
\]
This holds for each $J\in\left\{  1,\ldots,N_{i}\right\}  $, so%
\[
\mathbb{E}_{i}(N_{i}^{\prime})\leq\min_{J\in\left\{  1,\ldots,N_{i}\right\}
}\left(  J+(N_{i}-J)\left(  1-\mathbb{P}_{i}(\hat{A}_{i}^{J}\in C_{i}%
^{J})\right)  \right)  .
\]
This bound for $\mathbb{E}_{i}(N_{i}^{\prime})$ can be evaluated precisely. We
use computer-assisted computations for this purpose, first to compute the
probabilities~$\mathbb{P}_{i}(\hat{A}_{i}^{J}\in C_{i}^{J})$ and then the
minimum. For the Python code, see \cite{CondExp}. The results are presented in Figure~\ref{Fig:exp-bounds}, which shows that%
\[
\mathbb{E}_{i}(N_{i}^{\prime})<\frac{1}{6}N_{i}+40.
\]%
\begin{figure}[th]
\begin{center}
\includegraphics[
width=7cm,
]{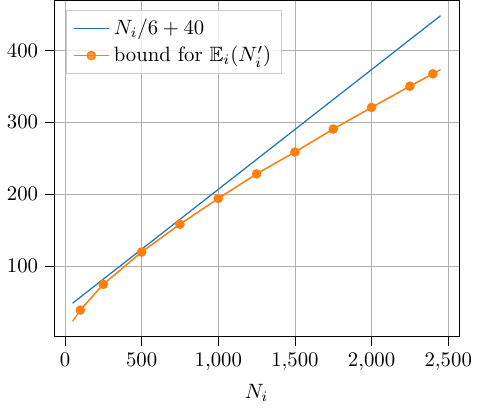}
\end{center}
\caption{The expectation $\mathbb{E}_{i}(N_{i}^{\prime})$ computed for various values of~$N_i$ versus the bound $\frac{1}{6}N_{i}+40$.}%
\label{Fig:exp-bounds}%
\end{figure}

Now we can obtain a bound for the expectation~$\mathbb{E}(M)$ of the
number~$M$ of applications of~$\mathbf{Q}$, where%
\[
M=\sum_{i=0}^{I}\frac{K_{i}-1}{2}N_{i}^{\prime}<\frac{1}{2}\sum_{i=0}^{I}%
K_{i}N_{i}^{\prime}.
\]
We have%
\begin{align*}
\mathbb{E}(M|I)  &  <\mathbb{E}\left(  \left.  \frac{1}{2}\sum_{i=0}^{I}%
K_{i}N_{i}^{\prime}\right\vert I\right)  =\mathbb{E}\left(  \left.  \frac
{1}{2}\sum_{i=0}^{I}K_{i}\mathbb{E}_{i}(N_{i}^{\prime})\right\vert I\right) \\
&  <\mathbb{E}\left(  \left.  \frac{1}{2}\sum_{i=0}^{I}K_{i}\left(  \frac
{1}{6}N_{i}+40\right)  \right\vert I\right)  =\mathbb{E}\left(  \left.
\frac{1}{12}\sum_{i=0}^{I}K_{i}N_{i}+20\sum_{i=0}^{I}K_{i}\right\vert
I\right)  ,
\end{align*}
so%
\[
\mathbb{E}(M)<\mathbb{E}\left(  \frac{1}{12}\sum_{i=0}^{I}K_{i}N_{i}%
+20\sum_{i=0}^{I}K_{i}\right)  .
\]
By recycling some of the estimates at the end of the proof of
Theorem~\ref{Thm:p9a3n7gd}, we obtain%
\begin{align*}
\frac{1}{2}\sum_{i=0}^{I}K_{i}N_{i}  &  <\frac{1}{\varepsilon}\left(  \frac
{1}{4E^{2}}\frac{\pi}{4}\left(  \frac{3}{2}\ln\left(  \frac{8}{C\alpha\pi}\right)  +\frac{3}{4}\ln\left(  3\right)  \right)  +\frac{3\pi}%
{16}\right)  ,\\
\sum_{i=0}^{I}K_{i}  &  <\sum_{i=0}^{I}3^{-I+i}\frac{\pi}{4}=\sum_{j=0}%
^{I}3^{-j}\frac{\pi}{4\varepsilon}<\frac{1}{\varepsilon}\frac{\pi}{6},
\end{align*}
so%
\begin{align*}
\mathbb{E}(M)  &  <\frac{1}{\varepsilon}\left(  \frac{1}{12}\left(  \frac
{1}{4E^{2}}\frac{\pi}{4}\left(  \frac{3}{2}\ln\left(  \frac{8}{C\alpha\pi}\right)  +\frac{3}{4}\ln\left(  3\right)  \right)  +\frac{3\pi}%
{16}\right)  +\frac{20\pi}{6}\right) \\
&  \approx\frac{1}{\varepsilon}\left(  27.380-10.201\ln\alpha\right)  .
\end{align*}
This bound for $\mathbb{E}(M)$ is shown by the line labelled as $\mathrm{AQAE}%
_{\mathbb{E}(M)}$ in Figures~\ref{Fig:anametf4d} and~\ref{Fig:aur1ba72}. As
expected, it is much tighter than the bound $\mathrm{AQAE}_{M}$ (hence even tighter than the
bounds $\mathrm{IQAE}_{M}$ and $\mathrm{MIQAE}_{M}$) for~$M$.

\end{document}